\documentclass{article}
\PassOptionsToPackage{round}{natbib}
\usepackage[preprint]{neurips_2026}

\usepackage[utf8]{inputenc} 
\usepackage[T1]{fontenc}    
\usepackage{url}            
\usepackage{nicefrac}       
\usepackage{microtype}      

\usepackage[american]{babel}

\usepackage{mathtools} 
\usepackage{booktabs} 
\usepackage{tikz} 

\usepackage{graphicx}
\usepackage{enumitem}
\usepackage{xcolor}

\usepackage{amssymb}
\usepackage{amsmath}
\usepackage{amsthm}
\usepackage{algorithm,algpseudocode}
\usepackage{dsfont}
\usepackage{bm}
\usepackage{booktabs}
\usepackage{hyperref}
\definecolor{mydarkblue}{rgb}{0,0.08,0.45}
\hypersetup{
  colorlinks=true,
  linkcolor=blue,
  citecolor=blue,
  urlcolor=purple
}
\usepackage[noabbrev,capitalize,nameinlink]{cleveref}
\usepackage{subfig}
\usepackage{wrapfig}
\usepackage{tikz}
\usetikzlibrary{shapes,arrows,automata,positioning}

\newcommand{\train}{\text{train}}
\newcommand{\val}{\text{val}}

\newcommand{\far}{\text{far}}

\DeclareMathOperator*{\argmin}{arg\,min}

\usepackage{stfloats} 

\newcommand{\algorithmicinput}{\textbf{Input:}}
\newcommand{\algorithmicoutput}{\textbf{Output:}}
\newcommand{\INPUT}{
\item[\algorithmicinput]}
\newcommand{\OUTPUT}{
\item[\algorithmicoutput]}

\usepackage{shadethm}
\usepackage{etoolbox}

\newtheoremstyle{new}
  {12pt}      
  {12pt}      
  {\itshape}  
  {}          
  {\bfseries\color{black}} 
  {.}         
  { }         
  {}          
\theoremstyle{new}
\newtheorem{theorem}{Theorem}[section]

\newtheorem{proposition}[theorem]{Proposition}
\newtheorem{lemma}[theorem]{Lemma}
\newtheorem{definition}[theorem]{Definition}

\AfterEndEnvironment{theorem}{\noindent\ignorespaces}
\AfterEndEnvironment{corollary}{\noindent\ignorespaces}
\AfterEndEnvironment{proposition}{\noindent\ignorespaces}
\AfterEndEnvironment{lemma}{\noindent\ignorespaces}
\AfterEndEnvironment{definition}{\noindent\ignorespaces}
\AfterEndEnvironment{assumption}{\noindent\ignorespaces}
\AfterEndEnvironment{example}{\noindent\ignorespaces}
\AfterEndEnvironment{remark}{\noindent\ignorespaces}
\AfterEndEnvironment{claim}{\noindent\ignorespaces}

\Crefname{theorem}{Theorem}{Theorems}
\Crefname{corollary}{Corollary}{Corollaries}
\Crefname{proposition}{Proposition}{Propositions}
\Crefname{lemma}{Lemma}{Lemmas}
\Crefname{definition}{Definition}{Definitions}
\Crefname{example}{Example}{Examples}
\Crefname{remark}{Remark}{Remarks}
\Crefname{claim}{Claim}{Claims}

\usepackage[breakable, theorems, skins]{tcolorbox}
\tcbset{enhanced}

\definecolor{shadethmcolor}{cmyk}{0,0,0,0.075}    
\definecolor{shaderulecolor}{rgb}{1,1,1}   

\newtheoremstyle{shad}
  {12pt}      
  {12pt}      
  {\itshape }  
  {}          
  {\bfseries\color{black}} 
  {.}         
  { }         
  {}          
\theoremstyle{shad} 

\newshadetheorem{thmbox}[theorem]{Theorem}
\newshadetheorem{corbox}[theorem]{Corollary}
\newshadetheorem{propbox}[theorem]{Proposition}
\newshadetheorem{lembox}[theorem]{Lemma}
\newshadetheorem{exbox}[theorem]{Example}
\newshadetheorem{defbox}[theorem]{Definition}
\newshadetheorem{rembox}[theorem]{Remark}

\newtheoremstyle{shad*}
  {12pt}      
  {12pt}      
  {\itshape }  
  {}          
  {\bfseries\color{black}} 
  {.}         
  { }         
  {}          
\theoremstyle{shad*} 
\newshadetheorem{anobox}{}

\Crefname{theorem}{Theorem}{Theorems}
\Crefname{corbox}{Korollar}{Corollaries}
\Crefname{propbox}{Proposition}{Propositions}
\Crefname{lembox}{Lemma}{Lemmas}
\Crefname{exbox}{Beispiel}{Examples}
\Crefname{defbox}{Definition}{Definitions}
\Crefname{rembox}{Anmerkung}{Remarks}

\usepackage{bm}
\usepackage{dsfont}

\newcommand{\ind}{\mathds{1}}
\newcommand{\R}{\mathds{R}}
\newcommand{\N}{\mathds{N}}

\newcommand{\E}{\mathds{E}}

\newcommand{\bu}{\bm{u}}

\newcommand{\bx}{\bm{x}}

\newcommand{\bz}{\bm{z}}

\newcommand{\bX}{\bm{X}}

\newcommand{\bZ}{\bm{Z}}

\newcommand{\Dcal}{\mathcal{D}}

\newcommand{\Ical}{\mathcal{I}}

\newcommand{\Scal}{\mathcal{S}}

\newcommand{\Vcal}{\mathcal{V}}
\newcommand{\Wcal}{\mathcal{W}}

\renewcommand{\bar}{\overline}

\providecommand{\Pr}{}
\renewcommand{\Pr}{\mathbb{P}}

\newcommand{\wh}[1]{\widehat{#1}}

\title{Throwing Vines at the Wall: Structure Learning via Random Search}
\author{%
  Thibault Vatter \\
  University of Applied Sciences Western Switzerland
  \And
  Thomas Nagler \\
  LMU Munich \\
  Munich Center for Machine Learning
}

\begin{document}
\maketitle

\begin{abstract}
  Vine copulas offer flexible multivariate dependence modeling and have become widely used in machine learning. Yet, structure learning remains a key challenge. Early heuristics, such as Dissmann's greedy algorithm, are still considered the gold standard but are often suboptimal. We propose random search algorithms  and a statistical framework based on model confidence sets, to improve structure selection, provide theoretical guarantees on selection probabilities and excess risk, as well as serve as a foundation for ensembling. Empirical results on real-world data sets show that our methods consistently outperform state-of-the-art approaches.
\end{abstract}

\section{Introduction}

Copulas provide a flexible framework for modeling complex multivariate distributions by disentangling marginal behavior from the dependence structure \citep[see, e.g.,][]{nelsen2007introduction,joe2014}.
Among them, the class of \emph{vine copulas} \citep{bedfordcooke2001,aasczadofrigessibakken2009} has become particularly popular in machine learning because it balances flexibility with tractability better than more classical copula families.
Recent applications span a wide range of areas, including domain adaptation \citep{Lopez-Paz2013}, variational inference \citep{Tran2015}, causal inference \citep{Lopez-Paz2016}, clustering \citep{tekumalla2017vine}, Bayesian optimization \citep{botied}, conformal prediction \citep{park2025semiparametric}, and generative modeling \citep{kulkarni2018nonparametric,Konstantelos2018,sun2019,Tagasovska2018b}.

A vine copula on $d$ variables consists of two components: its \emph{vine structure}, a nested sequence of undirected trees, and $d(d-1)/2$ bivariate (conditional) copulas, called \emph{pair-copulas}, indexed by its edges.
Vine copulas are most popular on problems from a handful to a few dozen variables. In this regime, the number of pair-copulas remains tractable, while the flexibility of the model can capture complex dependence patterns.
However, there exists $ 2^{(d-3)(d-2)/2 - 1}d!$ possible vines on $d$ variables \citep{morales2011count,joe2011regular}, which is of the same order as the number of directed acyclic graphs on $d$ nodes. This super-exponential growth makes exhaustive search infeasible beyond a few variables.
As a result, nearly all applications rely on greedy heuristics such as the spanning tree algorithm of \citet{Dissmann2013} or the forward selection procedure of \citet{kraus2017d}.

Despite their simplicity and lack of theoretical foundation, numerous attempts at improving on them had very limited success (see \cref{sec:related_work}).
The recent review of \citet{Czado2022} concludes that such heuristics still represent the state of the art, and that effective structure selection remains an important open challenge in vine methodology.

\paragraph{Contributions}
This paper challenges the prevailing view that the standard heuristics are difficult to improve upon. In particular:
\begin{itemize}
  \item We propose a vine structure learning algorithm based on hold-out random search and show that it outperforms state-of-the-art methods across a range of real-data settings for tasks ranging from generative modeling to regression.
  \item We integrate random search with a state-of-the-art algorithm for constructing \emph{model confidence sets} \citep{kim2025locally} tailored to vine structures, prove its validity in our context, and provide an efficient implementation. This allows to assess in-sample whether alternatives outperform a baseline heuristic and yields a principled set of competitive models.
  \item We demonstrate that vine regression methods based on averaging over the model confidence set consistently outperform single-vine approaches, including the current state of the art.
\end{itemize}

\begin{wrapfigure}{r}{0.52\textwidth}
  \centering
  \vspace{-0.5cm}
  \includegraphics[width=0.52\textwidth]{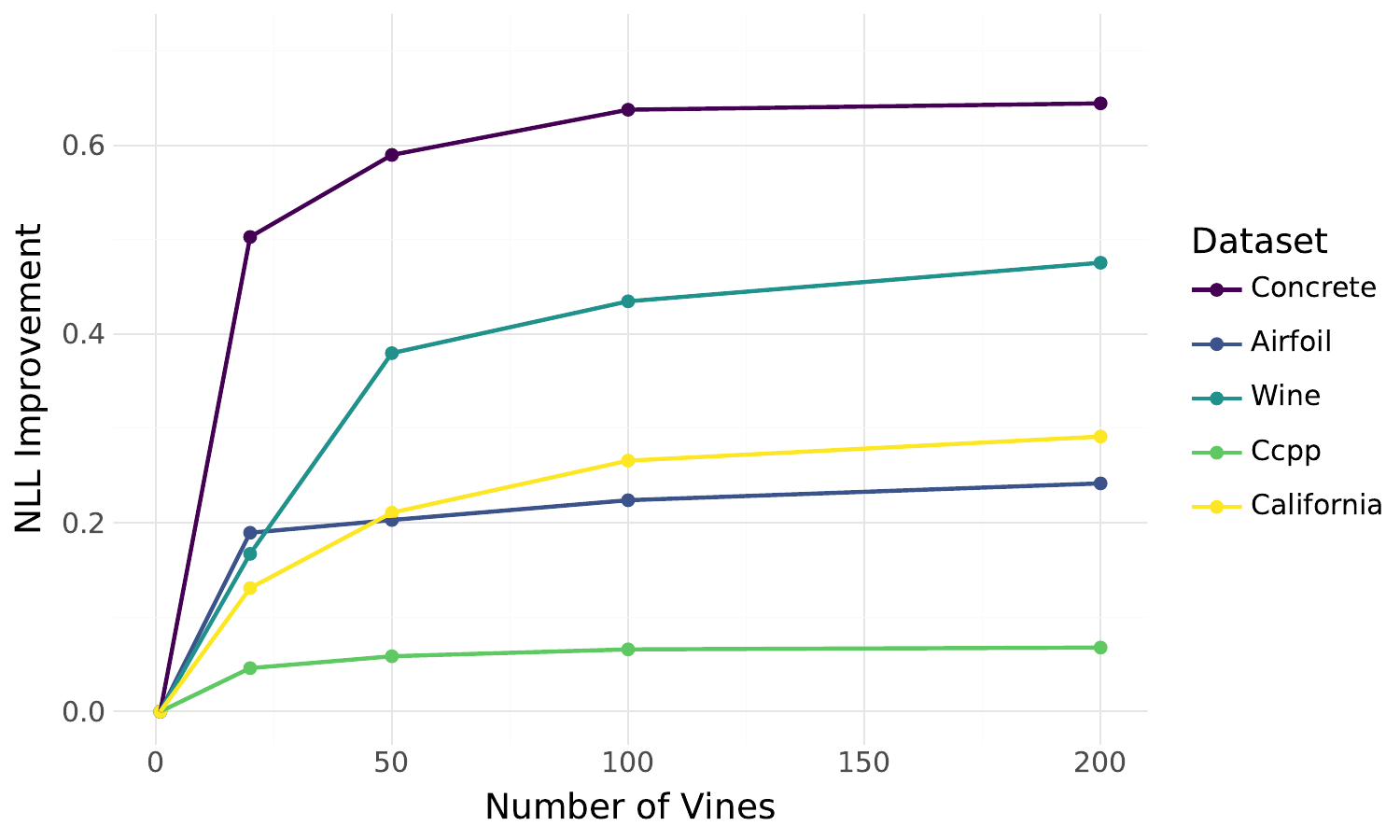}
  \caption{Average improvement in out-of-sample negative log-likelihood over \citet{Dissmann2013}.}
  \label{fig:intro-figure}
  \vspace{-1.0cm}
\end{wrapfigure}

In \cref{fig:intro-figure}, we illustrate the potential for improvement by comparing the out-of-sample negative log-likelihood of \citet{Dissmann2013} to that of vines sampled uniformly at random, selecting the best candidates using an independent validation set (see \cref{sec:experiments} for details).
From the picture, it is clear that the greedy heuristic is suboptimal, and that random search performance improves with the number of candidates.

The new algorithms are conceptually simple, straightforward to implement, and provide immediate benefits for applications of vine copulas in machine learning and related fields.

\paragraph{Outline}
In \cref{sec:background}, we review necessary background on vine copulas and structure learning.
In \cref{sec:methodology}, we introduce our new structure learning algorithms.
In \cref{sec:experiments}, we benchmark them against state-of-the-art approaches using real-world datasets.
We conclude in \cref{sec:conclusion}.

\section{Background} \label{sec:background}

\subsection{Copulas and vines}

Mathematically, a copula is a multivariate distribution with standard uniform margins.
From \citet{sklar1959}, a random vector $\bX \in \R^d$ has joint distribution $F$ with univariate marginal distributions $F_{1}, \dots, F_d$, if and only if there exists a copula $C$ such that, for each $x\in\R^d$,
\begin{align*}
  F\left(x_1,\,\ldots,\,x_d \right) = C\left\{ F_{1}(x_1),\,\ldots,\,F_{d}(x_d) \right\},
\end{align*}
which is unique if the margins are continuous.
Taking partial derivatives on both sides of this equation, the joint density is
the product of marginal densities $f_j = \partial F_j / \partial x_j$ and the \emph{copula density} $c = \partial^d C/\partial u_1 \cdots \partial u_d$:
\begin{align*}
  f(x_1, \dots, x_d) = c\left(F_1(x_1), \dots, F_d(x_d)\right) \times \textstyle\prod_{j = 1}^d f_j(x_j).
\end{align*}
Taking logarithm on both sides further implies that the joint log-likelihood is the sum of the marginal and copula log-likelihoods,
which can be exploited in a two-step estimation procedure \citep{genest1995semiparametric, joe1996estimation}:
first estimate each of the marginal distributions, and then the copula based on the \emph{pseudo-observations} $\wh U_i = \wh F_i(X_i)$, $i = 1, \dots, d$.


Following the seminal work of \citet{joe1996} and \citet{bedfordcooke2001, Bedford2002}, any copula density $c$ can be decomposed into a product of $d(d-1)/2$ bivariate (conditional) copula densities.
The order of conditioning in this decomposition can be organized using the following graphical structure.

\begin{definition} \label{def:vine_structure}
  A \emph{vine} is a sequence of spanning trees $(V_t, E_t)$ with $V_t$ and $E_t$ respectively the nodes and edges for $t = 1, \dots, d-1$, satisfying the following conditions: (i) $V_1=\{1, \dots, d\}$, (ii) $V_t=E_{t-1}$ for $t \ge 2$, and (iii) if two nodes in $(V_{t+1}, E_{t+1})$ are joined by an edge, the corresponding edges in $(V_t, E_t)$ must share a common node.
\end{definition}

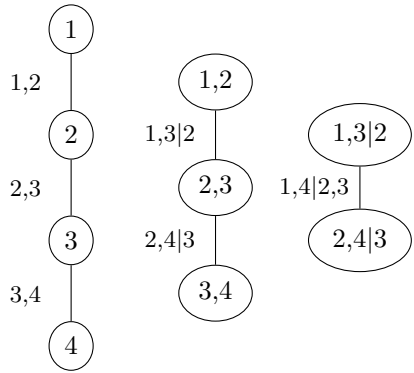
\begin{wrapfigure}{r}{0.4\textwidth}
  \vspace{-0.4cm}
  \centering
  \tikzstyle{VineNode} = [ellipse, fill = white, draw = black, text = black, align = center, minimum height = 0.1cm, minimum width = 0.1cm]
  \tikzstyle{DummyNode}  = [draw = none, fill = none, text = white]
  \newcommand{\xshiftLabela}{1.2cm}
  \newcommand{\yshiftLabela}{0.0cm}
  \newcommand{\xshiftLabelb}{-0.6cm}
  \newcommand{\yshiftLabelb}{-0.4cm}

  \resizebox{4.6cm}{3.6cm}{
    \begin{tikzpicture} [every node/.style = VineNode, node distance =0.7cm]
      \node (1){1}
      node[DummyNode] (D1-2) [below of = 1]{}
      node  (2)   [below of = D1-2]{2}
      node[DummyNode] (D2-3) [below of = 2]{}
      node  (3)   [below of = D2-3]{3}
      node[DummyNode] (D3-4) [below of = 3]{}
      node  (4)   [below of = D3-4]{4}

      node (1-2)  [right of = D1-2, xshift = \xshiftLabela, yshift = \yshiftLabela]{1,2}
      node[DummyNode] (D1-3_2) [right of = 2, xshift = \xshiftLabela, yshift = \yshiftLabela]{}
      node  (2-3)   [below of = D1-3_2]{2,3}
      node[DummyNode] (D2-4_3) [below of = 2-3]{}
      node  (3-4)   [right of = D3-4, xshift = \xshiftLabela, yshift = \yshiftLabela]{3,4}

      node (1-3_2) [right of = D1-3_2, xshift = \xshiftLabela, yshift = \yshiftLabela]{1,3$\vert$2}
      node[DummyNode] (D1-4_2-3) [right of = 2, xshift = \xshiftLabela, yshift = \yshiftLabela]{}
      node (2-4_3) [right of = D2-4_3, xshift = \xshiftLabela, yshift = \yshiftLabela]{2,4$\vert$3};

      \draw (1) to node[draw=none, fill = none, font = \footnotesize,
      above, xshift = \xshiftLabelb, yshift = \yshiftLabelb] {1,2} (2);
      \draw (2) to node[draw=none, fill = none, font = \footnotesize,
      above, xshift = \xshiftLabelb, yshift = \yshiftLabelb] {2,3} (3);
      \draw (3) to node[draw=none, fill = none, font = \footnotesize,
      above, xshift = \xshiftLabelb, yshift = \yshiftLabelb] {3,4} (4);
      \draw (1-2) to node[draw=none, fill = none, font = \footnotesize,
      above, xshift = \xshiftLabelb, yshift = \yshiftLabelb] {1,3$\vert$2} (2-3);
      \draw (2-3) to node[draw=none, fill = none, font = \footnotesize,
      above, xshift = \xshiftLabelb, yshift = \yshiftLabelb] {2,4$\vert$3} (3-4);
      \draw (1-3_2) to node[draw=none, fill = none, font = \footnotesize,
      above, xshift = \xshiftLabelb, yshift = \yshiftLabelb] {1,4$\vert$2,3} (2-4_3);

    \end{tikzpicture}
  }
  \caption{A vine structure on four variables.}
  \label{Vines:RVine_fig}
  \vspace{-0.5cm}
  \end{wrapfigure}
A simple example of a vine tree sequence is shown in \cref{Vines:RVine_fig}.
A \emph{vine copula} identifies each edge $e$ of the vine structure, with a label $\{j_e, k_e | D_e\}$ and a bivariate copula $C_{j_e, k_e| D_e}$, called a \emph{pair-copula}.
The sets $\{j_e, k_e \} \subset \{1, \ldots, d\}$ and $D_e \subset \{1, \ldots, d\}\setminus \{j_e, k_e\}$ are called the \emph{conditioned set} and \emph{conditioning set}, respectively, and we refer to \citet{Czado2022} for a precise definition.
Any copula density then factorizes as
\begin{align} \label{eq:vine_density}
  c(\bm u) & = \prod_{t=1}^{d-1} \prod_{e \in E_t} c_{j_e, k_e| D_e} \bigl(u_{j_e|D_e}, u_{k_e|D_e} \mid  \bm u_{D_e} \bigr),
\end{align}
with $u_{j_e|D_e} = C_{j_e|D_e}(u_{j_e}|\bm u_{D_e})$, $C_{j_e|D_e}$ the conditional distribution of $U_{j_e}$ given $\bm U_{D_e}$, and similarly for $u_{k_e|D_e}$.
As an example, decomposing the Gaussian copula density as in \eqref{eq:vine_density} leads to Gaussian pair-copulas $c_{j_e, k_e |D_e}$ with parameter equal to the partial correlation coefficient $\rho_{j_e, k_e ; D_e}$.

Thanks to the proximity condition, $u_{j_e|D_e}$ for any edge $e \in E_{t + 1}$ can be computed recursively from the pair-copulas in trees $1, \dots, t$.
This way, the parameters of vine copulas can be estimated by iterating between parameter estimation and conditional distribution evaluations, going from the first tree to the last \citep{aasczadofrigessibakken2009}. A similar decomposition holds when some variables are discrete, effectively by replacing derivatives of the distribution function by finite differences \citep{panagiotelis2012, funk2025towards}.

\subsection{Why the Structure Matters}

The key to the practical success of vine copulas is a \emph{simplifying assumption} that the conditional copulas in \eqref{eq:vine_density} do not depend on the conditioning value $\bu_{D_e}$.
This reduces the model to a collection of bivariate copulas, which are straightforward to estimate, yet still highly flexible. Each pair-copula can be chosen independently, allowing complex dependence patterns with varying strength, asymmetry, and non-linear tail behavior. More in-depth discussions of the assumption can be found in \citet{stoeber2013simplified, spanhel2019simplified, nagler2025simplified}.


The simplifying assumption also elevates the importance of the vine structure: different structures can yield substantially different approximations of the true copula density. Finding a structure that fits the data well is therefore crucial, both for predictive performance and for interpretability.

\subsection{Related Work} \label{sec:related_work}

Several previous works have addressed structure learning in vine copula models.

\paragraph{General-purpose methods}
The standard approach was introduced by \citet{Dissmann2013}, whose greedy algorithm constructs a maximum spanning tree based on Kendall's $\tau$. Variants using information criteria or goodness-of-fit $p$-values \citep{czado2013selection} offer no clear benefits but incur higher computational cost.
Attempts to incorporate tests of the simplifying assumption \citep{Kraus2017} also failed to consistently improve performance. More elaborate search strategies based on MCMC \citep{gruber2015sequential,gruber2018} or neural networks \citep{sun2019} can yield gains but are prohibitively expensive: even in the low-dimensional examples provided in the papers, the algorithm run times are on the order, or larger than, those of exhaustive search.
Randomized search variants have also been explored. Sampling from common variable orders \citep{zhu2020common} does not improve out-of-sample performance on real data and
the Monte Carlo tree search approach with lookahead of \citet{chang19a} only applies to truncated Gaussian models. Overall, the greedy heuristic of \citet{Dissmann2013} remains the default in practice \citep{Czado2022}.

\paragraph{Regression problems}
Specialized algorithms have been developed for regression tasks. These methods restrict attention to structures that admit closed-form conditional distributions, often star- or path-shaped trees (D- and C-vines). Some extend Dissmann’s approach \citep{chang2019prediction, chang19a, cooke2019vine}, while others use forward variable selection with early stopping \citep{kraus2017d, Schallhorn2017, zhu2021simplified, tepegjozova2022nonparametric}. Their applicability is thus narrower than general-purpose methods.

\paragraph{Sparse models}
Another line of work focuses on sparsity, i.e., identifying independence copulas among pairs, which becomes necessary in higher dimensions to keep models tractable.
Approaches include selecting truncation or thresholding levels \citep{kurowicka2011optimal, BrechmannCzadoAas2012, brechmann2015truncation, joe2018parsimonious, nagler2019model}, typically combined with spanning-tree heuristics. A different perspective exploits connections between vine copulas and Gaussian DAGs \citep{mueller2018representing,muller2019dependence,muller2019selection}, aiming primarily to detect conditional independencies rather than fully optimizing the vine structure. These questions are important but orthogonal to the focus of this paper, which is the structure itself, rather than the pair-copulas.

\section{Methodology} \label{sec:methodology}

\subsection{Setup} \label{sec:setup}

Let $\Dcal = \{\bm Z_i : 1 \leq i \leq n\}$ be a dataset of $n$ i.i.d.\ samples in $d$ dimensions, with $\bm Z_i = (Z_{i,1}, \ldots, Z_{i,d}) \in \R^d$.
In generative modeling tasks, these are the features of interest, while in regression-like problems we consider $\bm Z = (\bX, Y)$ as a feature--label pair.
Given a vine structure $\Vcal$ and a dataset $\Dcal$, we assume access to an algorithm that fits an estimated joint density $\wh f_{\Vcal, \Dcal} \in \mathcal{F}$, with $\mathcal{F}$ the space of probability densities in $\R^d$.
This density is composed of estimated marginal distributions and a simplified vine copula with structure $\Vcal$.

To evaluate different models, let $L\colon \mathcal{F} \times \R^d \to \R$ denote a loss functional, where $L(f, \bm Z)$ quantifies the quality of $f \in \mathcal{F}$ based on an observation $\bm Z$.
Common choices include the negative log-likelihood $L(f, \bm Z) = -\log f(\bm Z)$ for generative modeling, or similarly $L(f, (\bX, Y)) = -\log f(Y|\bX)$ for regression-type problems, and we refer to \citet{gneiting2014probabilistic} for a comprehensive review in the context of probabilistic forecasting.

The goal of model selection is to find, among a set of candidates, the one with minimal expected loss $R(f) = \E[L(f, \bm Z)]$.
In practice however, $R(f)$ is unknown and must be estimated from data, which can be done, e.g., via hold-out validation or cross-validation \citep{stone1974cross}.
In the remainder of the paper, all expectations and probabilities are
understood conditionally on the training set. Specifically, after splitting
$\Dcal$ into disjoint sets $\Dcal_\train$ and $\Dcal_\val$, we write
\begin{align*}
  \E'[\cdot]
  =
  \E[\cdot\mid\Dcal_\train],
  \quad
  \Pr'(\cdot)
  =
  \Pr(\cdot\mid\Dcal_\train), \quad R'(f)
  =
  \E'[L(f,\bm Z)], \quad R'(\Vcal) = R'(\wh f_{\Vcal, \Dcal_\train}).
\end{align*}

\subsection{Random Search Algorithm} \label{sec:RS}

\begin{algorithm}[th]
  \caption{Hold-out Random Search for Vines.}
  \label{alg:RS}
  \begin{algorithmic}[1]
    \INPUT Data set $\Dcal$ of size $n$, hold-out ratio $\eta$, number of candidate models $M$.
    \OUTPUT Optimal vine structure $\wh \Vcal$.\smallskip
    \State Split $\Dcal$ into training and validation sets:
    $\Dcal_\train \cup \Dcal_\val=\Dcal$, $n_{\val} = \lfloor\eta n\rfloor$, $n_{\train} = n - n_{\val}$.
    \State Generate candidate structures $\Theta = \{\Vcal_{1}, \dots, \Vcal_{M}\}$.
    \State For each candidate $\Vcal \in \Theta$:\smallskip
    \State \hspace*{12pt} Use the training set to fit a density $\wh f_{\Vcal, \Dcal_\train}$.
    \State \hspace*{12pt} Compute the validation loss
    $
      R_\val(\Vcal) = \frac{1}{|\Dcal_{\val}|} \sum_{\bm z \in \Dcal_{\val}} L(\wh f_{\Vcal, \Dcal_\train}, \bm z).
    $
    \State Select the structure with minimal validation loss:
    $
      \wh \Vcal = \argmin_{\Vcal \in \Theta} R_\val(\Vcal).
    $
    \end{algorithmic}
\end{algorithm}

To identify a suitable vine structure, we propose a simple and easy-to-implement algorithm based on hold-out random search, summarized in \cref{alg:RS}.
The procedure consists of splitting the data into training and validation sets, generating candidate structures at random, fitting them on the training data, and selecting the one with the smallest validation loss.

For a typical vine model $\wh f_{\Vcal, \Dcal}$, both training and inference incurs an $O(d^2)$ cost per observation, where the factor $d^2$ comes from the quadratic number of pair-copulas in the model.
This leads to an overall runtime complexity of $O(M n d^2)$ for \cref{alg:RS}.
Note that it is embarrassingly parallel over the $M$ candidates, as well as, for a single candidate, tree-wise over pair-copulas. Furthermore, one can  reduce the dependence on $d$ from quadratic to linear by truncating the vine at a fixed level \citep[e.g.,][]{BrechmannCzadoAas2012}.

For generating vine structures, we propose to sample uniformly at random using the algorithm of \citet{joe2011regular}; see also \citet[Section 6.13]{joe2014} and its implementation in the software libraries \texttt{VineCopula} and \texttt{pyvinecopulib} \citep{vinecopula2025, pyvinecopulib}. Other sampling schemes, potentially involving the training data, could be considered as well and also fit the theoretical framework and results ahead; see \cref{sec:wilson} for a sampling scheme based on local perturbations of the \cite{Dissmann2013} structure.

While the approach naturally extends to cross-validation, we argue that the additional computational budget is better invested in evaluating more random candidates; see  \cref{fig:intro-figure} and \cref{sec:experiments}.
In addition, cross-validation does not solve the problem that the selected model may be statistically indistinguishable from other candidates, which we address in the next subsection.

\subsection{Vine Confidence Sets}

On any given data set, we may want to assess whether the model found by \cref{alg:RS} is better than a benchmark.
When this is not the case, we might prefer the benchmark model for qualitative reasons such as ease of interpretability or easier comparison to previous analyses. And even if the benchmark is outperformed, multiple candidates from the random search may be statistically indistinguishable and used for model averaging.
We tackle this problem via \emph{model confidence sets (MCS)} \citep{hansen2011model}; roughly speaking, subsets of the candidates that contain the ``best ones'' with high probability.

Define the set of optimal models according to expected out-of-sample loss as
\[
  \Theta^* = \argmin_{\Vcal \in \Theta} R'(\Vcal),
\]
explicitly allowing for the possibility of multiple ones that are equally good.
We emphasize that, in our setup, the goal is not to identify a fixed ground truth, but to find structures that lead to trained models with minimal out-of-sample loss.
The optimal set $\Theta^*$ is itself random as determined by the training data, estimation algorithm, and candidate set.
In particular, we do not assume that the data was generated from a vine copula model.
An MCS is then another random subset that contains each of the optimal models with high probability, as formalized in the following definition.
\begin{definition}\label{def:MCS}
  An $\alpha$-level MCS is a random set $\wh \Theta \subseteq \Theta$ with $\Pr'(\Vcal \in \wh \Theta) \ge 1 - \alpha$, $\forall \Vcal \in \Theta^*$.
\end{definition}
Such an MCS offers marginal guarantees, in the sense that the probabilistic statement applies to each optimal model separately, rather than to the set as a whole.
An alternative concept is a uniform MCS requiring that  $\Theta^* \subseteq \wh \Theta$ with high probability. While this is a stronger guarantee, it is more difficult to achieve in practice, and may lead to unnecessarily large sets including many suboptimal models.


\subsection{A Vine MCS Algorithm}

Several methods for estimating MCSs have been proposed in the literature \citep{hansen2011model, lei2020cross, mogstad2024inference, zhang2024winners, kissel2022black}.
They generally rely on discrete argmin inference, where the null hypothesis for the optimality of a given candidate $\Vcal \in \Theta$ can be written as
$
  H_{0,\Vcal}\colon \; R'(\Vcal) \le \min_{\Wcal \in \Theta} R'(\Wcal).
$
If $\widehat T_{\Vcal}$ is a statistic satisfying $\widehat T_{\Vcal} \to N(0,1)$ under $H_{0,\Vcal}$, a natural MCS estimator obtains by inverting the test, that is $\wh \Theta
=\{\Vcal : \widehat T_{\Vcal} \le \Phi^{-1}(1-\alpha)\,\}$,
where $\Phi$ is the standard normal distribution function.
In this work, we use the DA-test of \citet{kim2025locally} because of its computational simplicity and good empirical performance.
A full description of the test can be found in \cref{sec:mcs_details} for completeness, along with a runtime comparison with alternative methods.

\begin{algorithm}[th]
  \caption{MCS for Vines.}
  \label{alg:MCS}
  \begin{algorithmic}[1]
    \INPUT Data set $\Dcal$, hold-out ratio $\eta$, number of candidate models $M$, confidence parameter $\alpha$.
    \OUTPUT Model confidence set of vine structure $\wh \Theta$.
    \State Follow Steps 1-3 from \cref{alg:RS}.
    \State Compute losses for $1 \le i \le n_\val$ and $1 \le m \le M$
    $L_{i, m} = L(\wh f_{\Vcal_m, \Dcal_{\train}}, \bZ_i)$, $\bZ_i \in \Dcal_\val$.
    \State Run the MCS algorithm on the scores $L_{i, m}$, yielding an $\alpha$-MCS of indices $\wh \Ical \subseteq \{1, \dots, M\}$.
    \State Return the vine MCS $\wh \Theta = \{\Vcal_m\colon m \in \wh \Ical\}.$
  \end{algorithmic}
\end{algorithm}

\cref{alg:MCS} describes how this MCS method fits into our vine structure learning pipeline. The algorithm has runtime complexity $O(M n d^2 + M n)$, where the first term is from the first three steps (see the discussion in \cref{sec:RS}), the second from the MCS procedure, and we refer to \cref{sec:mcs_complexity} for this part of the complexity analysis.
It also comes with the following theoretical guarantees.
\begin{proposition} \label{prop:mcs_marginal}
  Whenever
  $
    \limsup_{n \to \infty}\max_{\Vcal \in \Theta}\E'[| L(\wh f_{\Vcal, \Dcal_\train}, \bm Z) |^3] < \infty,
  $
  \cref{alg:MCS} yields \\
  (i) $\liminf_{n \to \infty}  \Pr'(\Vcal \in \wh \Theta) \ge 1 - \alpha$ for all $ \Vcal \in \Theta^*$,  \\
  (ii) $  \lim_{n \to \infty}  \Pr'\bigl(\exists \Vcal \in \wh \Theta\colon R'(\Vcal)-R'(\Vcal^{*}) > a_n/  \sqrt{ n}\bigr) \to 0$ for any  $\Vcal^* \in \Theta^*$ and $a_n \to \infty$.
\end{proposition}
The proof is given in \cref{sec:proofs}.
Part (i) states that $\wh \Theta$ is indeed an approximate $\alpha-$level MCS.
Part (ii) implies that $\wh \Theta$ excludes, with probability tending to one, all candidates whose excess risk exceeds a statistically detectable threshold $a_n/\sqrt n$. Since $a_n$ may diverge arbitrarily slowly, this threshold vanishes, so any model falling below it are asymptotically risk-equivalent to the optimal ones.
Both guarantees are asymptotic and conditional on the training data, which determines the optimal set $\Theta^*$.
The condition on the third moment of the loss is mild and typically satisfied in practice.

\subsection{MCS Ensembles} \label{sec:mcs_ensembles}

Usually, an MCS contains multiple models, whose performance is effectively indistinguishable given the limited amount of observations. In such cases, it is natural to take an ensemble approach. This may have the additional benefit of variance reduction and improved predictive performance, as is well-known from the literature on model averaging \citep{hoeting1999bayesian, burnham2002model}.
In particular, we propose to define the MCS mixture model
\begin{align*}
  \wh f_{\wh \Theta}(\bz) = \frac{1}{|\wh \Theta|} \sum_{\Vcal \in \wh \Theta} \wh f_{\Vcal,\mathcal{D}}(\bz),
\end{align*}
which is straightforward to use in both downstream generative and regression-type tasks.
Additionally, for any loss that is convex in $f$, the risk of the ensemble itself is no worse than the average risk of the individual models in the MCS, and is close to the risk of the best one, i.e.,
\begin{align*}
  R(\wh f_{\wh\Theta})
  \le
  \frac{1}{|\wh\Theta|}
  \sum_{\Vcal\in\wh\Theta}
  R'(\Vcal)
  \le
  R'(\Vcal^{*})+
  O_{\Pr}(a_n/\sqrt n),
\end{align*}
with the first inequality by Jensen's inequality and the second by \cref{prop:mcs_marginal}.

For regression problems, we adapt the estimating equation approach of \citet{Nagler2018} to the MCS mixture model. Specifically, let $\psi_\beta$ be a function such that $\E[\psi_\beta(Y) \mid \bX = \bx] = 0$ if and only if $\beta$ is the target of inference. Examples are $\psi_\beta(y) = y - \beta$ for the conditional mean $\beta(\bx) = \E[Y \mid \bX = \bx]$ or $\psi_\beta(y) = \ind\{y < \beta(\bx)\} - \tau$ for the conditional $\tau$-quantile.
We now approximate the conditional expectation by the fitted model, i.e., we solve
\begin{align}\label{eq:estimating_equation}
  \int \psi_{\beta} \bigl( y \bigr) \wh f_{\wh \Theta}(y \mid \bx) dy = 0.
\end{align}
In practice, we consider a set of equally spaced grid points $\{y_1, \dots, y_G\}$. If \eqref{eq:estimating_equation} holds, we also have
$\sum_{g=1}^G \psi_{\beta}(y_g) \wh f_Y(y_g) \sum_{\Vcal \in \wh \Theta} \wh c_{\Vcal, \mathcal{D}}(\wh F_Y(\bx), \wh F_Y(y_g))\approx 0,
$
which can easily be solved numerically.
This approach is convenient because it puts no restrictions on the structure of the vine copula, facilitates model averaging, and can be used for all sorts of conditional quantities.
We refer to \cref{sec:comp_details} and \cref{sec:impl_details} for additional discussions and computational details.


\section{Experiments} \label{sec:experiments}

To validate the new methods, we provide a comprehensive comparison using popular real-world data sets on density estimation and regression tasks.

\subsection{Setup}

\paragraph{Data sets}

We use five data sets from the UCI repository \citep{Dua:2019}, \texttt{Energy Efficiency} ($n=768$, $p=8$), \texttt{Concrete Compressive Strength} ($n=1030$, $p=8$), \texttt{Airfoil Self-Noise} ($n=1503$, $p=5$),   \texttt{Wine quality} ($n=4898$, $p=12$), \texttt{Combined Cycle Power Plant (Ccpp)} ($n=9568$, $p=5$), as well as the \texttt{California Housing} dataset ($n=20640$, $p=8$) \citep{Pace1997}. All have a continuous label variables and are commonly used for benchmarking density estimation and tabular regression models.
The number of features is in the typical range for applications of vine copula models, see also the discussions in  \cref{sec:related_work} and \cref{sec:conclusion} for additional challenges in higher-dimensional settings.
For all datasets, we standardize the features to have mean zero and unit variance, randomly split the data into 80\% training, and 20\% test sets.
For methods requiring a validation set, we further split the training
set into 75\% training and 25\% validation.

\paragraph{Computational details} \label{sec:comp_details}

All experiments were run on a single workstation equipped with an AMD Ryzen Threadripper 2990WX 32-core processor and 125 GiB of RAM; no cluster or GPU resources were used.
A \texttt{Python} package implementing the proposed methods and scripts to reproduce the experiments are available in the supplementary materials and described in \cref{sec:impl_details}.
It relies mostly on the \texttt{pyvinecopulib} package \citep{pyvinecopulib} for vine copula modeling.
Marginal distributions are estimated by kernel local-likelihood and the pair-copulas using the nonparametric \texttt{TLL} estimator, which combines speed and flexibility \citep{nagler2017nonparametric}.

For MCS inference, we rely on a \texttt{Python} translation of the \texttt{R} \citep{R2025} script gracefully provided by the authors of \citet{kim2025locally}, improved to scale as $O(M n)$ instead of $O(M^2 n)$; see \cref{sec:mcs_complexity} for details.
All experiments were repeated 10 times with different random seeds, and we report the average performance and standard error across these replications.

\subsection{Density Estimation}

In this experiment, we assess the quality of density estimation by stacking the features and the label, so that the problem becomes $(p+1)$-dimensional.
We compare the following methods:
\begin{itemize}
  \item \texttt{Dissmann}: A first benchmark method using the heuristic of \citet{Dissmann2013}, based on absolute Kendall's $\tau$.
  \item \texttt{Kraus}: A second benchmark method using the heuristic of \citet{Kraus2017}, based on a combination of Kendall's $\tau$ and $p$-values of tests for the simplifying assumption.
  \item \texttt{RS-B (M)}: Our random search \cref{alg:RS} with $M$ candidates, using the best model on the validation set with $L(f, \bm Z) = -\log f(\bm Z)$.
  \item \texttt{RS-E (M)}: An ensemble over the MCS obtained by \cref{alg:MCS} with $M$ candidates with the same loss as \texttt{RS-B} and $\alpha = 0.05$ (see \cref{sec:mcs_ensembles}).
\end{itemize}
Other approaches mentioned in \cref{sec:related_work} are not included in the comparison because of their prohibitive computational cost or lack of general applicability.

\begin{table*}[th]\small
  \caption{Rounded test NLL (standard error) for density estimation. Best results in bold.}
  \label{tab:NLL}
  \begin{tabular*}{\linewidth}{@{\extracolsep{\fill}}lllllll}
    \toprule
    & Energy & Concrete & Airfoil & Wine & Ccpp & California \\
    \midrule\addlinespace[2.5pt]
    \texttt{Dissmann}   & 1.95 (0.35) & 7.14 (0.11) & 3.25 (0.04) & 8.49 (0.19) & 6.80 (0.01)  & 3.63 (0.16) \\
    \texttt{Kraus}      & 1.51 (0.32) & 7.16 (0.12) & 3.27 (0.04) & 8.39 (0.15) & 6.80 (0.01)  & 3.60 (0.17)  \\
    \texttt{RS-B (50)}  & 0.47 (0.36) & 7.08 (0.11) & 3.19 (0.06) & 8.51 (0.19) & 6.77 (0.01) & 3.49 (0.17) \\
    \texttt{RS-B (100)} & 0.36 (0.36) & 7.05 (0.14) & 3.20 (0.06) & 8.51 (0.19) & 6.76 (0.01) & 3.47 (0.16) \\
    \texttt{RS-B (500)} & 0.20 (0.37) & 7.04 (0.11) & 3.09 (0.03) & 8.51 (0.19) & 6.75 (0.01) & 3.45 (0.17) \\
    \texttt{RS-E (50)}  & 0.12 (0.43) & 6.55 (0.10) & 3.05 (0.05) & 8.11 (0.18) & 6.74 (0.01) & 3.41 (0.16) \\
    \texttt{RS-E (100)} & 0.13 (0.34) & 6.51 (0.09) & 3.03 (0.04) & 8.06 (0.19) & 6.74 (0.01) & 3.36 (0.16) \\
    \texttt{RS-E (500)} & \textbf{-0.28 (0.32)} & \textbf{6.49 (0.09)} & \textbf{3.00 (0.04)} & \textbf{7.92 (0.18)} & \textbf{6.73 (0.01)} & \textbf{3.31 (0.17)} \\
    \bottomrule
  \end{tabular*}
\end{table*}

In \cref{tab:NLL}, we show the average negative log-likelihood (NLL) on test data.
Except for the wine data set, both of our random search approaches outperform the benchmarks, irrespective of the number of candidates $M$.
The performance improves with $M$, as expected, and our \texttt{RS-E (500)} method is the best on all data sets, albeit the margin is small in in some cases.
The results also explain why the data set Energy is not included in \cref{fig:intro-figure}: the performance gain over the benchmark is simply too large to be displayed on the same scale.

\begin{wrapfigure}{r}{0.5\textwidth}
  \vspace{-0.3cm}
  \centering
  \includegraphics[width=\linewidth]{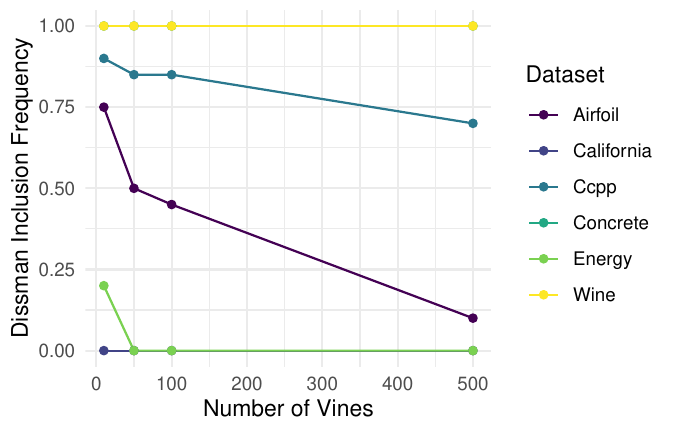}
  \caption{Density estimation: how often the \texttt{Dissmann} method lies in $95\%$-MCS (computed on the training set), estimated over 20 seeds.}
  \label{fig:dissman_in_mcs}
  \vspace{-0.4cm}
\end{wrapfigure}
A key benefit of the MCS approach is that it enables assessing whether a benchmark is statistically indistinguishable from the best candidates using only the training data. 
\cref{fig:dissman_in_mcs} reports the frequency with which \texttt{Dissmann} belongs to the $95\%$-MCS across datasets, which closely mirrors the results in \cref{tab:NLL}.
Notably, for the only dataset where \texttt{Dissmann} is better than \texttt{RS-B} (Wine), it is included in the MCS in every replication. In practice, this provides a simple check: the benchmark can be retained when it is competitive, without risking performance loss when it is clearly inferior.

As mentioned in \cref{sec:RS}, different structure sampling schemes could be considered. We provide such an example in \cref{sec:wilson}, where sampling is based on local perturbations of the \cite{Dissmann2013} structure. The results are qualitatively similar to \texttt{RS-B} and \texttt{RS-E}, albeit with the local perturbation approach slightly underperforming compared to the uniform sampling scheme.

\subsection{Mean and Median Regression}

In this experiment, we compare the following methods:
\begin{itemize}
  \item \texttt{Dissmann}: The same heuristic as above, combined with \citet{Nagler2018} to turn the density estimate into a conditional mean or median regression function (see \cref{sec:mcs_ensembles}).
  \item \texttt{Kraus}: The method of \citet{kraus2017d}, which uses the fact that D-Vines admit closed-form conditional distributions.
  \item \texttt{RS-B (M)}: Our random search \cref{alg:RS} with $M$ candidates, using the best model on $\Dcal_{\text{val}}$ with $L(f, \bm Z) = -\log f(Y | \bm X)$, turned into a regression function as for \texttt{Dissmann}.
  \item \texttt{RS-E (M)}: Same, albeit for \cref{alg:MCS} with the MCS using $\alpha = 0.05$, and \citet{Nagler2018} adapted to ensembles as described in \cref{sec:mcs_ensembles}.
\end{itemize}

\begin{table*}[th]\small
  \caption{Rounded test RMSE (standard error) for mean regression. Best results in bold.}
  \label{tab:RMSE}
  \begin{tabular*}{\linewidth}{@{\extracolsep{\fill}}lllllll}
    \toprule
    & Energy & Concrete & Airfoil & Wine & Ccpp & California \\
    \midrule\addlinespace[2.5pt]
    \texttt{Dissmann} & 2.76 (0.10) & 7.12 (0.13) & 4.52 (0.08) & 0.63 (0.01) & 4.09 (0.03) & 0.66 (0.00) \\
    \texttt{Kraus} & 3.4 (0.13) & 7.09 (0.14) & 4.48 (0.10) & 0.63 (0.01) & 4.26 (0.02) & 0.65 (0.00) \\
    \texttt{RS-B (50)} & 2.55 (0.23) & 6.98 (0.13) & 3.73 (0.06) & 0.64 (0.01) & 4.08 (0.03) & 0.62 (0.00) \\
    \texttt{RS-B (100)} & 2.33 (0.20) & 7.02 (0.15) & 3.76 (0.07) & 0.64 (0.01) & 4.07 (0.03) & 0.62 (0.00) \\
    \texttt{RS-B (500)} & 2.06 (0.13) & 6.9 (0.16) & 3.73 (0.07) & 0.63 (0.01) & 4.08 (0.03) & 0.62 (0.00) \\
    \texttt{RS-E (50)} & 2.00 (0.06) & 6.28 (0.10) & \textbf{3.73 (0.06)} & 0.61 (0.01) & 4.07 (0.03) & 0.60 (0.00) \\
    \texttt{RS-E (100)} & 2.04 (0.07) & 6.24 (0.10) & 3.78 (0.06) & 0.61 (0.00) & 4.06 (0.03) & 0.60 (0.00) \\
    \texttt{RS-E (500)} & \textbf{1.87 (0.06)} & \textbf{6.24 (0.08)} & 3.79 (0.07) & \textbf{0.61 (0.01)} & \textbf{4.06 (0.03)} & \textbf{0.60 (0.00)} \\
    \bottomrule
  \end{tabular*}
\end{table*}

In \cref{tab:RMSE}, we show the root mean squared error (RMSE) on test data for mean regression.
Similarly to the density estimation experiment, random search methods outperform the benchmarks in most cases, with performance improving with $M$.
Again, \texttt{RS-E (500)} is the best method in five out of six datasets, with the RMSE of \texttt{RS-E (M)} being consistently lower than that of \texttt{RS-B (M)} for almost all $M$ and datasets.
The exact same pattern is observed for median regression when evaluated via mean absolute error (MAE); see \cref{sec:additional_experiments} for details.
This suggests that ensembling via the MCS is particularly beneficial in regression-type problems, where the model selection is not directly optimized for the evaluation metric.

\subsection{Probabilistic Forecasting}

In this experiment, we compare the same methods as in the regression experiment.
However, we evaluate the quality of the full predictive distribution via the continuous ranked probability score (CRPS) \citep{gneiting2014probabilistic}, a popular proper scoring rule for probabilistic forecasting.
Specifically, we use the fact that the methods from the previous experiment yield quantiles at all levels, and compute the $\text{CRPS}(F, y) = 2 \int_0^1 \rho_{\tau}(y - F^{-1}(\tau)) d\tau$,
with $\rho_\tau(u) = u(\tau - \ind\{u < 0\})$ the check function.
In \cref{tab:CRPS}, we see that outperformance of our random search methods over the benchmarks is even stronger in this experiment, and ensembles via the MCS are particularly effective.

\begin{table*}[th]\small
  \vspace{-0.3cm}
  \caption{Rounded test CRPS (standard error) for probabilistic forecasting. Best results in bold.}
  \label{tab:CRPS}
  \begin{tabular*}{\linewidth}{@{\extracolsep{\fill}}lllllll}
    \toprule
    & Energy & Concrete & Airfoil & Wine & Ccpp & California \\
    \midrule\addlinespace[2.5pt]
    \texttt{Dissmann} & 1.36 (0.03) & 3.91 (0.07) & 2.38 (0.04) & 0.31 (0.00) & 2.25 (0.01) & 0.33 (0.00) \\
    \texttt{Kraus} & 1.41 (0.05) & 3.94 (0.06) & 2.37 (0.04) & 0.31 (0.00) & 2.38 (0.01) & 0.33 (0.00) \\
    \texttt{RS-B (50)} & 1.14 (0.07) & 3.81 (0.08) & 1.95 (0.02) & 0.31 (0.00) & 2.24 (0.01) & 0.31 (0.00) \\
    \texttt{RS-B (100)} & 1.12 (0.10) & 3.80 (0.07) & 1.97 (0.03) & 0.31 (0.00) & 2.23 (0.01) & 0.31 (0.00) \\
    \texttt{RS-B (500)} & 0.97 (0.05) & 3.78 (0.08) & 1.94 (0.03) & 0.30 (0.00) & 2.24 (0.01) & 0.30 (0.00) \\
    \texttt{RS-E (50)} & 0.93 (0.02) & 3.39 (0.05) & \textbf{1.94 (0.03)} & 0.29 (0.00) & 2.24 (0.01) & 0.29 (0.00) \\
    \texttt{RS-E (100)} & 0.95 (0.02) & 3.37 (0.05) & 1.97 (0.03) & 0.29 (0.00) & 2.24 (0.01) & 0.29 (0.00) \\
    \texttt{RS-E (500)} & \textbf{0.89 (0.02)} & \textbf{3.37 (0.04)} & 1.97 (0.03) & \textbf{0.29 (0.00)} & \textbf{2.23 (0.01)} & \textbf{0.29 (0.00)} \\
    \bottomrule
  \end{tabular*}
\vspace{-0.3cm}
\end{table*}

\subsection{Runtime}

\begin{wrapfigure}{r}{0.54\textwidth}
  \vspace{-0.6cm}
  \centering
  \includegraphics[width=\linewidth]{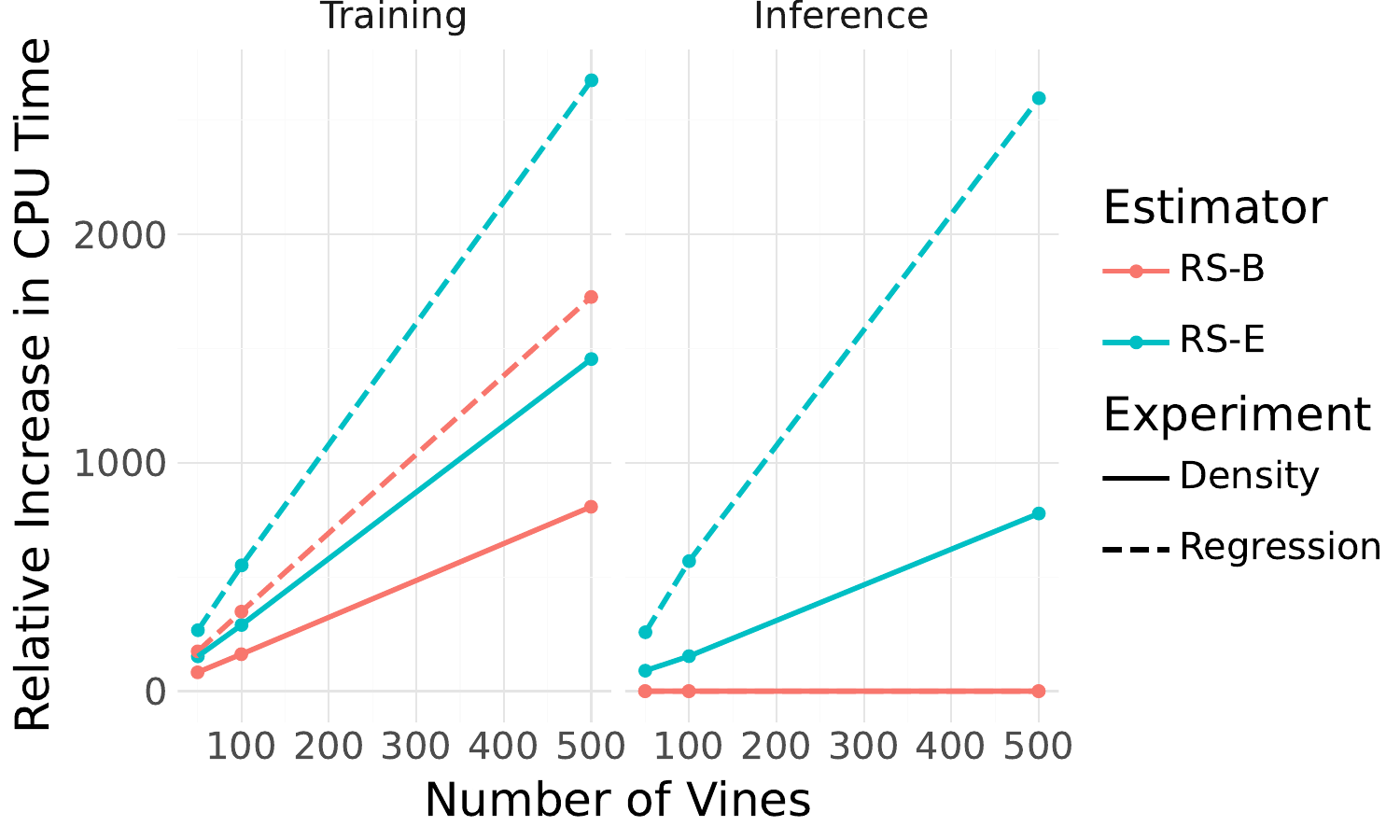}
  \caption{Relative CPU time of random search methods compared to \texttt{Dissmann} on the \texttt{Concrete} data set.}
  \label{fig:runtime}
  \vspace{-0.5cm}
\end{wrapfigure}

In \cref{fig:runtime}, we show the relative CPU time of our random search methods compared to \texttt{Dissmann} on the \texttt{Concrete} data set.
As expected, the CPU time grows roughly linearly with $M$ in training, irrespective of the random search method or the task.
On the other hand, there is no increase in inference time for \texttt{RS-B}, since only one model is used for prediction.
For \texttt{RS-E}, the inference time grows with $M$, since some fraction of the candidates are retained in the MCS and used for prediction.
    
While the computational cost of the random search methods is higher than that of \texttt{Dissmann}, the total cost is still tiny for real datasets of this scale. For instance, training and inference average respectively 0.357 and 0.006 seconds on a single CPU core for the \texttt{Concrete} data set. Additionally, the random search methods are easily parallelizable, since candidates can be fitted independently.
See \cref{sec:additional_experiments} for additional runtime results on other data sets.

\section{Conclusion} \label{sec:conclusion}

We propose a simple and yet effective method for learning vine copula structures based on random search and model confidence sets. Our experiments on real-world data sets demonstrate that the new method outperforms existing greedy heuristics in both density estimation and regression tasks, often by a large margin.

\paragraph{Limitations} The main limitation of the method is the additional computational cost compared to greedy heuristics. However, this cost is modest if the dimension is not too high and easily parallelizable, making it a worthwhile trade-off for the significant performance gains achieved. Nonetheless, in truly high-dimensional problems, any structure learning method must be combined with additional sparsity-inducing mechanisms, such as truncation, regularization, or variable selection.
While beyond the scope of the current work, designing more sophisticated sampling methods that incorporate such mechanisms is an important and promising direction for future research.

A second limitation is that our theoretical guarantees are stated for a fixed
number of candidates $M$ and dimensions $d$. Allowing both to grow with $n$ should be possible using dimension-agnostic discrete argmin inference arguments similar to those in \cite{kim2025locally}. Since the focus of this paper is methodological
and empirical rather than theoretical, and the considered regime already covers the main range of current vine copula
applications, we leave this
extension to future work.

\begin{ack}
  We thank Johanna F. Ziegel for pointing us to the literature on model confidence sets.
  This work was partially supported by the Hasler Stiftung under grant number 2025-05-01-519.
\end{ack}

\bibliography{bibliography}

\appendix
%

%

\onecolumn


\section{Details on the MCS Algorithm} \label{sec:mcs_details}

In the following, we recall the relevant details of the MCS algorithms given by \citet{kim2025locally} to keep our paper self-contained.
Recall the notation of our main paper, specifically \Cref{alg:MCS}: There are $M$ candidate models $\{\Vcal_m\}_{m=1}^M$ and we have already computed validation losses $L_{i,m}$ for $1 \le i \le n_{\val}$ and $1 \le m \le M$. For ease of notation, define the loss differences $\Delta_{m,k}^{(i)} = L_{i,m} - L_{i,k}$ for $1 \le i \le n_{\val}$ and $1 \le m,k \le M$.

\subsection{The DA Test Statistic and  \texttt{DA-MCS-marg} Construction} \label{sec:mcs_marg}

We first describe the discrete argmin (DA) hypothesis test introduced by \citet{kim2025locally} which serves as the foundation for constructing our MCS. The basic algorithm to compute the test statistic from a set of loss observations is given in \Cref{alg:DAtest}.

\numberwithin{algorithm}{section}

\begin{algorithm}[h]
  \caption{DA test statistic}
  \label{alg:DAtest}
  \begin{algorithmic}[1]
    \INPUT Loss observations $\{L_{i,m}\colon 1 \le i \le N,\, 1 \le m \le M\}$, target index $r \in \{1,\dots,M\}$.
    \OUTPUT Test statistic $\wh T_r$.\medskip

    \State Define a split of indices: $I_1 = \{1,\dots,\lfloor N/2\rfloor\}$ and $I_2 = \{\lfloor N/2\rfloor+1,\dots,N\}$. \medskip
    \State On $I_1$, find the alternative model that minimizes the cumulative validation loss:
      $
        \wh m_{-r} \in \argmin_{m \in \{1,\dots,M\}\setminus\{r\}} \sum_{i \in I_1} L_{i,m},
     $
      breaking ties by the smallest index.  \medskip
    \State On $I_2$, compute
      $
          \wh T_r \;=\; \frac{1}{\wh \sigma_{r,\wh m_{-r}}} \cdot \frac{1}{\sqrt{|I_2|}} \sum_{i \in I_2} \Delta_{r,\wh m_{-r}}^{(i)},
      $
      where $\wh \sigma_{r,\wh m_{-r}}^2$ is the sample variance of $\{\Delta_{r,\wh m_{-r}}^{(i)}: i \in I_2\}$.
  \end{algorithmic}
\end{algorithm}

Under the null hypothesis that model $r$ has the smallest expected out-of-sample loss among the $M$ candidates, a one-sided test rejects for large positive values of $\wh T_r$, e.g. when $\wh T_r > \Phi^{-1}(1-\alpha)$, where $\Phi$ is the standard normal CDF.


To obtain an MCS with marginal coverage guarantees, one can simply invert the DA test: compute $\wh T_r$ for all $r=1,\dots,M$ by running \Cref{alg:DAtest} and define
\begin{align*}
  \wh \Theta = \left\{r \in \{1, \dots, M\} \colon \; \wh T_r \le \Phi^{-1}(1 - \alpha)\right\}.
\end{align*}
This procedure, which we call \texttt{DA-MCS-marg}, yields an MCS with marginal coverage $1 - \alpha$ under mild conditions, as formally stated in \cref{prop:mcs_marginal}.

\subsection{Efficient Implementation and Complexity} \label{sec:mcs_complexity}

When implemented exactly as described above, the \texttt{DA-MCS-marg} procedure has complexity $O(n_{\val} M^2)$:
\begin{itemize}
  \item Each call to \cref{alg:DAtest} requires $O(n_{\val} M)$ time to compute the cumulative losses, $O(M)$ time to find the argmin, and $O(n_{\val})$ time to compute the sample variance and test statistic. This results in a total of $O(n_{\val} M)$ time per call.
  \item  Because \cref{alg:DAtest} is called $M$ times in total for \texttt{DA-MCS-marg}, the overall complexity is $O(n_{\val} M^2)$
\end{itemize}
However, this complexity can be reduced to $O(n_{\val} M)$ for both procedures by avoiding redundant computation when looping over $r \in \{1, \dots, M\}$:
\begin{itemize}
  \item The column sums $\bar L_m = \sum_{i \in I_1} L_{i,m}$ for $1 \le m \le M$ can be computed once in $O(n_{\val} M)$ time and accessed in $O(1)$ time in each call to \cref{alg:DAtest}.
  \item The argmins $\wh m_{-1}, \dots, \wh m_{-M}$, can be computed simultaneously in $O(M)$ time. Specifically, one pass through the precomputed column sums $\bar L_1, \dots, \bar L_M$ suffices to find the two smallest numbers $\bar L_{j_1}$ and $\bar L_{j_2}$ among $\{\bar L_1, \dots, \bar L_M\}$
  (breaking ties by index) along with corresponding indices $j_1$ and $j_2$. Then, use another pass to set $\wh m_{-r} = j_1$ if $\bar L_r \neq \bar L_{j_1}$ and $\wh m_{-r} = j_2$ otherwise, for all $r=1,\dots,M$.
\end{itemize}

\cref{fig:mcs_naive_vs_fast} compares the average computation time of the naive $O(n_{\val} M^2)$ implementation and the fast $O(n_{\val} M)$ implementation of \texttt{DA-MCS-marg} (as well as a variant for obtaining uniform MCS) on a log-log scale.

\begin{figure}[!t]
  \centering
  \includegraphics[width=0.85\textwidth]{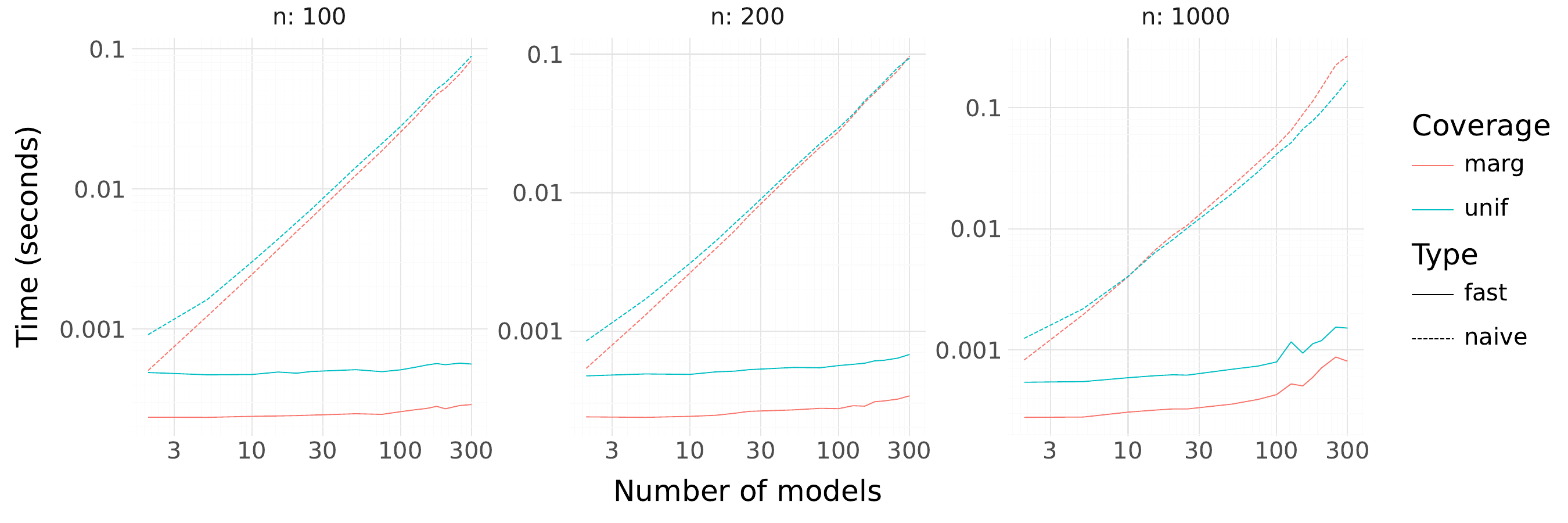}
  \caption{Average computation time of different DA-MCS Python implementations.}
  \label{fig:mcs_naive_vs_fast}
\end{figure}

\subsection{Runtime Comparison of MCS Methods}

In \cref{fig:mcs_comparison}, we compare the average computation time of different MCS methods.
Method names are from the R script provided by the authors of \citet{kim2025locally}; \texttt{KR-PLG} refers to \texttt{DA-MCS-marg}.
We observe that the method is significantly faster than competing methods.

\begin{figure}[!t]
  \centering
  \includegraphics[width=0.85\textwidth]{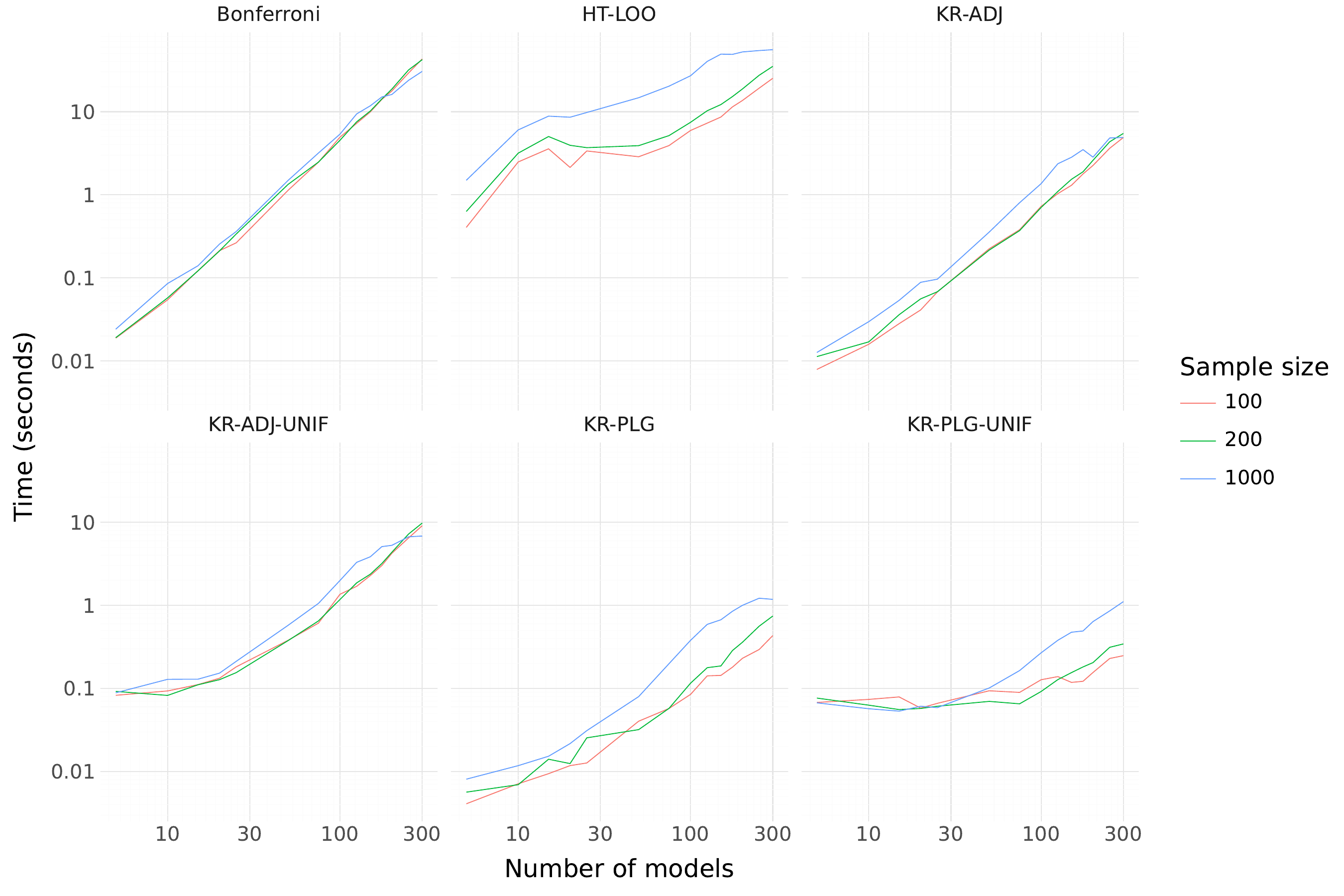}
  \caption{Average computation time of MCS methods.}
  \label{fig:mcs_comparison}
\end{figure}

\section{Proofs of the Theoretical Results} \label{sec:proofs}

The proof uses similar arguments as \citet{kim2025locally}, but requires some technical modifications to adapt to our setting.
We start with a small lemma simplifying the moment condition required by \citet{kim2025locally}.
To simplify notation, we write $\E'[\cdot] = \E[\cdot \mid \Dcal_\train]$ and $\Pr'(\cdot) = \Pr(\cdot \mid \Dcal_\train)$ throughout this section.

\begin{lemma} \label{lem:moment}
  Let $M \in \N$ and define $\Delta_{m, k} = L(\wh f_{\Vcal_m, \Dcal_\train}, \bm Z) - L(\wh f_{\Vcal_k, \Dcal_\train}, \bm Z)$ and $\bar \Delta_{m, k} = \Delta_{m, k} - \E'[\Delta_{m, k}]$.
  If
  \begin{align} \label{eq:moment_cond}
    \limsup_{n \to \infty}\max_{1 \le m \le M}\E'[| L(\wh f_{\Vcal_m, \Dcal_\train}, \bm Z) |^3 ] < \infty
  \end{align}
  there is $K \in (0, \infty)$ such that for all $n$ large enough,
  \begin{align} \label{eq:toshow}
    \max_{1 \le m, k \le M} \frac{\E'[|\bar \Delta_{m, k}|^3]}{\E'[|\bar \Delta_{m, k}|^2]^{3/2}} \le K,
  \end{align}
  with the convention $0 / 0 = 1$.
\end{lemma}

\begin{proof}[Proof of \Cref{lem:moment}]
  The moment condition \eqref{eq:moment_cond} implies that there are $C \in (0, \infty)$ and $n_0 \in \N$ such that for all $n \ge n_0$,
  \begin{align*}
    \max_{1 \le m \le M} \E'[| L(\wh f_{\Vcal_m, \Dcal_\train}, \bm Z) |^3] \le C.
  \end{align*}
  Suppose $n \ge n_0$ from now on. It holds
  \begin{align*}
    \E'[|\bar \Delta_{m, k}|^3] & \le \E'[(| \Delta_{m, k}| + |\E'[ \Delta_{m, k}]|)^3]
    \le 4 \E'[|\Delta_{m, k}|^3] + 4 |\E'[ \Delta_{m, k}]|^3
    \le 8 \E'[|\Delta_{m, k}|^3],
  \end{align*}
  using triangle inequality in the first, $(a + b)^3 \le 4(|a|^3 + |b|^3)$ in the second, and Jensen's inequality in the last step. Using the same arguments, we also have
  \begin{align*}
    \E'[| \Delta_{m, k}|^3] & \le 4 \E'[| L(\wh f_{\Vcal_m, \Dcal_\train}, \bm Z)|^3] +  4 \E'[| L(\wh f_{\Vcal_k, \Dcal_\train}, \bm Z)|^3] \le 8C.
  \end{align*}
  Overall, we have shown that $\max_{1 \le m, k \le M}\E'[|\bar \Delta_{m, k}|^3] \le 64C$.
  Now let
  \begin{align*}
    \Scal = \left\{(m, k) \in \{1, \dots, M\}^2\colon E'[|\bar \Delta_{m, k}|^2] > 0 \right\},
  \end{align*}
  $\Scal^c = \{1, \dots, M\}^2 \setminus \Scal$,
  and define $c = \min_{(m, k) \in \Scal} E'[|\bar \Delta_{m, k}|^2]^{3/2}$, noting that $c > 0$. We have
  \begin{align*}
    \max_{1 \le m, k \le M} \frac{\E'[|\bar \Delta_{m, k}|^3]}{\E'[|\bar \Delta_{m, k}|^2]^{3/2}}
     & \le \max_{(m, k) \in \Scal} \frac{\E'[|\bar \Delta_{m, k}|^3]}{\E'[|\bar \Delta_{m, k}|^2]^{3/2}} + \max_{(m, k) \in \Scal^c} \frac{\E'[|\bar \Delta_{m, k}|^3]}{\E'[|\bar \Delta_{m, k}|^2]^{3/2}} \\
     & \le 64 C/ c + \max_{(m, k) \in \Scal^c} \frac{\E'[|\bar \Delta_{m, k}|^3]}{\E'[|\bar \Delta_{m, k}|^2]^{3/2}} .
  \end{align*}
  Recall that $(m, k) \in \Scal^c$ implies $\E'[|\bar \Delta_{m, k}|^2] = 0$. This further implies $\bar \Delta_{m, k} = 0$ almost surely, and because $\E'[|\bar \Delta_{m, k}|^3] < \infty$, it follows that $\E'[|\bar \Delta_{m, k}|^3] = 0$. Altogether have shown that
  \begin{align*}
    \max_{1 \le m, k \le M} \frac{\E'[|\bar \Delta_{m, k}|^3]}{\E'[|\bar \Delta_{m, k}|^2]^{3/2}} \le 64 C/ c + 1 =: K < \infty,
  \end{align*}
  which completes the proof.
\end{proof}

\begin{proof}[Proof of \Cref{prop:mcs_marginal} (i)]
    Let $r \in \Theta^*$ be arbitrary but fixed.  
    Define
    \begin{align*}
         \wh T_{r}^* = \frac{1}{\wh \sigma_{r,\wh m_{-r}}} \cdot \frac{1}{\sqrt{|I_2|}} \sum_{i \in I_2} \left( \Delta_{r,\wh m_{-r}}^{(i)} -  \E'[\Delta_{r,\wh m_{-r}}^{(i)} \mid \wh m_{-r}] \right),
    \end{align*}
    a centered version of  $\wh T_r$ that satisfies $\wh T_r \le \wh T_r^*$ since $r \in \Theta^*$ implies $\E'[\Delta_{r,\wh m_{-r}}^{(i)} \mid \wh  m_{-r}] \le 0$. It holds
    \begin{align*}
      \Pr'(r \in \wh \Theta \mid \wh m_{-r}) & = \Pr'(\wh T_{r} \le \Phi^{-1}(1 - \alpha) \mid \wh m_{-r}) \\
      &\ge \Pr'(\wh T_{r}^* \le \Phi^{-1}(1 - \alpha) \mid \wh m_{-r}) \\
      &\ge 1 - \alpha - \sup_{t \in \R}\left| \Pr'(\wh T_{r}^* \le t \mid \wh m_{-r}) - \Phi(t) \right|.
    \end{align*}
    To show that the supremum on the right is negligible, we apply the Berry-Esseen bound for Student's statistic of \citet{bentkus1996berry}, conditionally on $\Dcal_\train, \Vcal_1, \dots, \Vcal_M$, and $\wh m_{-r}$ (which are all independent of $\{L_{i, m}\colon n_\val / 2 < i \le n_\val \}$). This yields 
    \begin{align*}
      \sup_{t \in \R} \left| \Pr'(\wh T_{r}^* \le t \mid \wh m_{-r}) - \Phi(t) \right| \le \frac{C}{\sqrt{n_{\val}}} \frac{\E'[|\bar \Delta_{r,   \wh m_{-r}}|^3 \mid \wh m_{-r}]}{\E'[|\bar \Delta_{r,   \wh m_{-r}}|^2 \mid \wh m_{-r}]^{3/2}}  \le \frac{C}{\sqrt{\lfloor \eta   n\rfloor}} \max_{1 \le m, k \le M} \frac{\E'[|\bar \Delta_{m,  k}|^3]}{\E'[|\bar \Delta_{m, k}|^2]^{3/2}},
    \end{align*}
    for a universal constant $C \in (0, \infty)$.
    \cref{lem:moment} implies that the right-hand side is bounded by $CK / \sqrt{\lfloor \eta   n\rfloor}$ for all $n$ large enough. Altogether, this gives
    \begin{align*}
       \Pr'(r \in \wh \Theta) = \E'[\Pr'(r \in \wh \Theta \mid \wh m_{-r})] \ge 1 - \alpha - \frac{CK}{\sqrt{\lfloor \eta   n\rfloor}} =  1 - \alpha - o(1),
    \end{align*}
    as claimed.
\end{proof}

\begin{proof}[Proof of \Cref{prop:mcs_marginal} (ii)]
Let
\begin{align*}
  \wh R_1(\Vcal)
  =
  \frac{1}{|I_1|}
  \sum_{i\in I_1}
  L(\wh f_{\Vcal,\Dcal_\train},\bZ_i) ,\qquad \wh R_2(\Vcal)
  =
  \frac{1}{|I_2|}
  \sum_{i\in I_2}
  L(\wh f_{\Vcal,\Dcal_\train},\bZ_i),
  \end{align*}
  denote the validation risk on the selection and testing split, respectively, and by $\Wcal_{-\Vcal} \in \argmin_{\Wcal\in\Theta\setminus\{\Vcal\}} \wh R_1(\Wcal)$ the selected competitor of $\Vcal$.
  Further, define the set of ``far'' models as
  \begin{align*}
        \Theta_{\far}
     =
    \left\{\Vcal\notin\Theta^*: R'(\Vcal)-R'(\Vcal^{*})>\frac{a_n}{\sqrt{n}}\right\},
  \end{align*}
  where $a_n = O(\sqrt{n})$ is a sequence that can be chosen arbitrarily slowly diverging to $+\infty$.
    
  We first show that, uniformly over $\Vcal\in\Theta_{\far}$,
  $\widehat{\Wcal}_{-\Vcal}$ is near-oracle in terms of population risk. Since
  $\Vcal\in\Theta_{\far}$ implies $\Vcal\notin\Theta^*$, an oracle
  $\Vcal^*\in\Theta^*$ belongs to $\Theta\setminus\{\Vcal\}$. Hence, by
  definition of $\widehat{\Wcal}_{-\Vcal}$, $ \wh R_1(\widehat{\Wcal}_{-\Vcal})
  \le
  \wh R_1(\Vcal^*)$, and therefore,
  \begin{align*}
  R'(\widehat{\Wcal}_{-\Vcal}) - R'(\Vcal^*)
  &=
  \{R'(\widehat{\Wcal}_{-\Vcal})
  - \wh R_1(\widehat{\Wcal}_{-\Vcal})\}
  +
  \{\wh R_1(\widehat{\Wcal}_{-\Vcal})
  - \wh R_1(\Vcal^*)\}
  +
  \{\wh R_1(\Vcal^*) - R'(\Vcal^*)\} \\
  &\le
  2\max_{\Wcal\in\Theta}
  |\wh R_1(\Wcal)-R'(\Wcal)|.
  \end{align*}
Because $M = |\Theta|$ is fixed and the third moments of the losses are uniformly
bounded, Chebyshev's inequality and a union bound imply
\begin{align*} 
  \max_{\Wcal\in\Theta}
  |\wh R_1(\Wcal)-R'(\Wcal)|
  =
  O_{\Pr'}(n^{-1/2}) = o_{\Pr'}(a_n/\sqrt n),
  \end{align*}
  whence
  \begin{align} \label{eq:near_oracle}
    R'(\widehat{\Wcal}_{-\Vcal}) - R'(\Vcal^*) = o_{\Pr'}(a_n/\sqrt n).
  \end{align}
By definition, for $\Vcal\in\Theta_{\far}$,
\begin{align}\label{eq:far}
  R'(\Vcal)-R'(\Vcal^*) > \frac{a_n}{\sqrt n}.
\end{align}
Combining \eqref{eq:near_oracle} and \eqref{eq:far} yields, uniformly over
$\Vcal\in\Theta_{\far}$,
  \begin{align*}
    R'(\Vcal)-R'(\widehat{\Wcal}_{-\Vcal})
    >
    \frac{a_n}{\sqrt n}\{1+o_{\Pr'}(1)\}.
  \end{align*}
  Now consider the corresponding DA statistic:
  \begin{align*}
    \wh T_{\mathrm{DA}}(\Vcal)
    &:=
    \frac{1}{\wh \sigma_{\Vcal, \widehat{\Wcal}_{-\Vcal}}} \times \sqrt{{|I_2|}} \times [R_2(\Vcal) - R_2(\widehat{\Wcal}_{-\Vcal})].
  \end{align*}
  Using the third moment bound, Markov's inequality, and the union bound (noting that $M$ is fixed), we get 
  and
  \begin{align*}
  \max_{\Vcal\in\Theta_{\far}}\wh \sigma_{\Vcal, \widehat{\Wcal}_{-\Vcal}} = O_{\Pr'}(1) \quad \text{and} \quad \max_{\Vcal\in\Theta_{\far}}
  |R_2(\Vcal) - R'(\Vcal)|
  =
  O_{\Pr'}(1 / \sqrt{n}).
\end{align*}
  Further, by the definition of $I_2$, it holds $|I_2| \ge n_{\val} / 2 \ge \lfloor \eta n / 2\rfloor$, so $1/\sqrt{|I_2|} = O(1/\sqrt{n})$.
  Combining these facts it holds, uniformly over $\Vcal\in\Theta_{\far}$,
  \begin{align*}
    \frac{1}{\wh T_{\mathrm{DA}}(\Vcal)}
    =
    \frac{\wh \sigma_{\Vcal, \widehat{\Wcal}_{-\Vcal}}}{\sqrt{|I_2| [R_2(\Vcal) - R_2(\widehat{\Wcal}_{-\Vcal})]}}
    =  \frac{O_{\Pr'}(1)}{\sqrt{n}[R'(\Vcal) - R'(\widehat{\Wcal}_{-\Vcal}) + O_{\Pr'}(1 / \sqrt{n})]}  
    = O_{\Pr'}\left(\frac{1}{a_n}\right).
  \end{align*}
  Since $a_n \to \infty$, this implies,
  \begin{align*}
    \max_{\Vcal\in\Theta_{\far}}
    \Pr'\left(\Vcal\in\wh\Theta\right) = \max_{\Vcal\in\Theta_{\far}}
    \Pr'\left( T_{\mathrm{DA}}(\Vcal) \le \Phi^{-1}(1 - \alpha)\right)
    \to 0.
  \end{align*}

  Because $M$ is fixed, a union bound gives
  \begin{align*}
      \Pr'\left(\exists \Vcal \in \wh\Theta \colon R'(\Vcal)-R'(\Vcal^{*}) > \frac{a_n}{\sqrt n} \right) =
     \Pr'\left(\wh\Theta\cap\Theta_{\far}\ne\emptyset\right)
    \le
    \sum_{\Vcal\in\Theta_{\far}}
    \Pr'\left(\Vcal\in\wh\Theta\right)
    \to 0,
  \end{align*}
  as claimed.
\end{proof}


\section{Implementation Details} \label{sec:impl_details}

To ensure both reproducibility and impact, we bundle the proposed methods into a Python package named \texttt{\textbf{vinesforests}}, whose estimators follow the \texttt{scikit-learn}~\citep{scikit-learn} API (\texttt{fit}, \texttt{predict}, \texttt{score}, \texttt{score\_samples}, etc.) and are compatible with \texttt{numpy}~\citep{numpy} arrays and, where indicated, \texttt{pandas}~\citep{pandas} \texttt{DataFrame}s. As part of the supplementary material, we provide a \texttt{zip} file containing the following file structure:

\begin{itemize}
  \item \texttt{README.md} — Instructions for installation and running experiments.
  \item \texttt{pyproject.toml} — Configuration file for the package, Pixi and dependencies.
  \item \texttt{src/} — Source code for the \texttt{vinesforests} package.
  \begin{itemize}
    \item \texttt{vine}: This module exposes two estimator classes, \texttt{VineDensity} and \texttt{VineRegressor}, which implement single vine density estimation and regression, respectively.
    \item \texttt{forest}: This module exposes two estimator classes, \texttt{VineForestDensity} and \texttt{VineForestRegressor}, which implement the ensemble methods for density estimation and regression, respectively.
    \item \texttt{benchmark}: This module exposes two estimator classes, \texttt{KrausVineDensity} and \texttt{KrausVineRegressor}, which are wrappers around R implementations of the baselines.
    \item \texttt{experiments}: This module contains utilities for running experiments, including a CLI interface.
    \item \texttt{mcs.py} — Implementation of the MCS procedures.
    \end{itemize}
  \item \texttt{tests/} — Unit tests.
  \item \texttt{notebooks/} — Scripts and Jupyter notebooks allowing to fully reproduce our experiments and analysis.
\end{itemize}

\subsection{Single vine estimators} The class \texttt{VineBase}, which derives from \texttt{sklearn.base.BaseEstimator}, centralizes preprocessing and vine--copula fitting using \texttt{pyvinecopulib}~\citep{pyvinecopulib}. Inputs may be \texttt{numpy} arrays or \texttt{pandas} \texttt{DataFrame}s containing numeric columns, ordered categoricals, and unordered categoricals. Unordered categoricals are expanded to dummy variables via a map $ x \mapsto \{\mathbf{1}(x = \ell_j)\}_{j>1} $, with dummies represented as ordered categoricals taking values $ \{0,1\} $.
Note that, for both \texttt{VineBase} and its subclasses \texttt{VineDensity} and \texttt{VineRegressor}, expensive computations are batch vectorized and processed in chunks of size \texttt{batch\_size} to trade memory for throughput.

\paragraph{Marginal distributions}
A per-feature univariate kernel density estimator, \texttt{pyvinecopulib.Kde1d}, is used for each marginal; its type is set to \texttt{"continuous"} or \texttt{"discrete"} according to the post-expansion schema. We refer to \href{https://vinecopulib.github.io/pyvinecopulib/examples/07_kde1d.html}{this notebook with examples} and \href{https://vinecopulib.github.io/pyvinecopulib/_generate/pyvinecopulib.Kde1d.html}{the \texttt{pyvinecopulib.Kde1d} documentation} for more details, notably on automatic bandwidth selection.
For a given model, the pseudo--observations used by the vine are computed as follows. For continuous variables, the probability integral transform $ U = F(Z) $ given by \texttt{pyvinecopulib.Kde1d.cdf} is used. For discrete features, both $ F(Z) $ and the left--limit $ F(Z^{-}) $ are computed and stacked.

\paragraph{Vine copula}
Given marginals $ \{F_j\}_j $, the vine copula $ C $ is fit on pseudo--observations $ U = (U_1,\ldots,U_d) $ with $ U_j = F_j(Z_j) $.
In the density context, $\bZ = \bX$ is the feature vector. 
In the regression context, $\bZ = (Y, \bX)$ is the joint vector of target and features, with $Y$ in the first dimension treated as continuous. 
For vine copulas, we use the \texttt{pyvinecopulib.Vinecop} class and refer to \href{https://vinecopulib.github.io/pyvinecopulib/_generate/pyvinecopulib.Vinecop.html}{the documentation} for details, as well as the \href{https://vinecopulib.github.io/pyvinecopulib/examples/02_vine_copulas.html}{notebook with examples}.
In \texttt{vinesforest}, a vine copula is estimated via \texttt{pyvinecopulib.Vinecop.from\_data(data=U, structure=..., var\_types=..., controls=...)}:
\begin{itemize}
  \item If a \texttt{structure} is not specified, \texttt{pyvinecopulib} performs structure selection using \citet{Dissmann2013} by default. Otherwise, it needs to be an instance of \href{https://vinecopulib.github.io/pyvinecopulib/_generate/pyvinecopulib.RVineStructure.html#pyvinecopulib.RVineStructure}{\texttt{pyvinecopulib.RVineStructure}} (e.g., generated via \texttt{pyvinecopulib.RVineStructure.simulate(d)}).
  \item By default, the controls are \path|pyvinecopulib.FitControlsVinecop(family_set=[pyvinecopulib.tll], num_threads=1)|, so that the \texttt{TLL} family is used for all pair-copulas. Other families can be specified via the \texttt{family\_set} argument, parallelism is controlled with \texttt{num\_threads}, and we refer to \href{https://vinecopulib.github.io/pyvinecopulib/_generate/pyvinecopulib.FitControlsVinecop.html}{the documentation} for more options.
\end{itemize}

This class provides a \texttt{\_pdf\_samples(X, y=None, log=False, copula\_only=False)} method which is reused by subclasses.
It implements the pipeline of computing pseudo--observations, evaluating the copula density via \texttt{pyvinecopulib.Vinecop.pdf}, and potentially multiplying by marginal densities using \texttt{pyvinecopulib.Kde1d.pdf}.



\paragraph{\texttt{VineDensity}} For this class, which derives from \texttt{VineBase}, the \texttt{fit} method calls the above pipeline on $ X $ only. The log-likelihood per sample is returned by \texttt{score\_samples(X)}, and the average by \texttt{score(X)}. Sampling draws $ U \sim C $ via \texttt{vine.simulate} and transforms back with the fitted marginals using $ X_j = F_j^{-1}(U_j) $.

\paragraph{\texttt{VineRegressor}}  For this class, which derives from \texttt{VineBase} as well, the \texttt{fit} fits the vine on the joint $ (Y, X) $. Predictions via \texttt{predict} use conditional importance weights derived from the copula density. For a test observation $ x $, the method computes
\begin{align*}
  w_i(x) \propto \begin{cases}
    c_{Y,X}(U_Y(y_i), U_X(x)),\quad i=1,\dots,n_{\text{train}} , & \text{if } \texttt{use\_grid = False}, \\
    c_{Y,X}(U_Y(y_i), U_X(x))  f_Y(y_i),\quad i=1,\dots,n_{\text{grid}} , & \text{if } \texttt{use\_grid = True},
  \end{cases}
\end{align*}
where the grid is a set of equally spaced grid points $\{y_1, \dots, y_G\}$ in $ Y $-space, and $ U_Y(y_i) = F_Y(y_i) $ is computed with the marginal CDF.
The second option is usually faster when the training set is large, as it reweights by the marginal density $ f_Y(y) $ to correct for non-uniformity of grid points instead of using the entire empirical distribution.
Conditional expectations are computed in closed form via $\sum_i w_i(x) y_i$, and conditional quantiles are computed by weighted quantiles of $ \{y_i\}_i $ using the ``inverted CDF'' definition.
For more details, we refer to the \href{https://numpy.org/devdocs/reference/generated/numpy.quantile.html}{\texttt{numpy.quantile}} definition, and in particular its \texttt{weights} argument.

To evaluate the conditional log-likelihood used in the forest selection (see below), \texttt{VineRegressor} provides a method \texttt{copula\_marginal\_density(X, log=False, n\_grid=101)}, which uses the identity
\[
  \log f_{Y\mid X}(y\mid x)
  = \log c_{Y,X}(U_Y(y), U_X(x)) - \log c_X(U_X(x)) + \log f_Y(y),
\]
where $ c_X(u_X) = \int_0^1 c_{Y,X}(u_Y, u_X)\,du_Y $ is approximated with Simpson's rule on a uniform grid over $ [0,1] $.

\subsection{Forest estimators (ensemble level)} 
Similarly as in the single vine case, a class \texttt{VineForestBase} deriving from \texttt{sklearn.base.BaseEstimator} implements structure randomization, bootstrap resampling, validation splitting, survivor (either via the MCS or simply those better than \citet{Dissmann2013} in the validation set) selection, and parallel fitting and prediction averaging.
It uses a base learner, which can be either \texttt{VineDensity} or \texttt{VineRegressor}, and has two subclasses \texttt{VineForestDensity} and \texttt{VineForestRegressor}.
After fitting the default estimator using \citep{Dissmann2013} for structure selection, for a requested $ M $ base learners, \texttt{VineForestBase} proceeds as follows:
\begin{enumerate}
  \item For each $ r=1,\ldots,M $, use \texttt{pyvinecopulib.RVineStructure.simulate} to simulate a random $ R $-vine structure $ \mathcal{V}_r $, optionally bootstraps the training rows, and fits a cloned base estimator with this structure.
  \item Scores each candidate on the held-out validation fold using per-sample log-likelihoods: $ \log f_X(x) $ for density and $ \log f_{Y\mid X}(y\mid x) $ for regression (as above).
  \item Selects survivors either by keeping all candidates strictly improving over the default or by running the MCS procedure at level $ \alpha $ (\texttt{da\_mcs\_marg}). Optionally \texttt{best\_only} retains the single best survivor.
  \item Refits all survivors on the full training set and averages predictions at inference time (log of mean densities for \texttt{VineForestDensity}; averaging conditional weights for \texttt{VineForestRegressor}).
\end{enumerate}
Note that:
\begin{itemize}
  \item Parallelism uses \texttt{joblib}~\citep{joblib}. When \texttt{n\_jobs} exceeds the number of survivors, the code increases the \texttt{pyvinecopulib} \texttt{num\_threads} of each base learner to utilize all cores.
  \item The selection routine \texttt{da\_mcs\_marg}  (module \texttt{mcs.py}) implement the procedures described in \cref{sec:mcs_details}, with the refinement described in \cref{sec:mcs_complexity} to make it scale linearly with the number of candidates.
  \item The two derived classes \texttt{VineForestDensity} and \texttt{VineForestRegressor} only implement respectively the \texttt{score/score\_sample}, and \texttt{predict} methods, which call the appropriate methods of the base learners and aggregate results.
\end{itemize}

\subsection{Other Components}

\paragraph{Determinism and Performance} All randomization is seeded through a user--provided integer seed; random vine structures derive from \texttt{numpy} RNG streams. Batch size controls the memory/runtime trade-off for likelihood evaluation and prediction. Vectorization is used throughout, and heavy loops are limited to batched calls into \texttt{pyvinecopulib}. Where feasible, computations are parallelized via \texttt{joblib} and intra-estimator threads in \texttt{pyvinecopulib}.

\paragraph{Benchmarks and R Interfacing} The classes \texttt{KrausVineRegressor} and \texttt{KrausVineDensity} are wrappers for the R packages \texttt{vinereg} and \texttt{kdevine}, respectively, through \texttt{rpy2}. A small mixin, \texttt{RInterface}, ensures silent R output, proper conversion of \texttt{pandas} categoricals to R factors, and robust conversion between \texttt{numpy}/\texttt{pandas} and R objects.
The two classes are also derived from \texttt{VineRegressor} and \texttt{VineDensity}, respectively, to ensure compatibility with the \texttt{scikit-learn} API and the rest of the codebase.

\paragraph{Experiments}
The file \texttt{utils.py} offers: (i) \texttt{expand\_factors} for categorical handling; (ii) \texttt{crps\_from\_quantiles} implementing Simpson--rule CRPS from predictive quantiles.
As for the file \texttt{real\_data.py}, it contains lightweight loaders for the UCI/OpenML datasets used in the paper and applies standardization where appropriate; when categorical variables are present (e.g., Bike Sharing), ordered categoricals are preserved for compatibility with \texttt{VineRegressor}.
The rest of the module contains a CLI interface to run experiments, which is used in the provided Jupyter notebooks to reproduce results.

\section{Additional Experiments} \label{sec:additional_experiments}

\subsection{Results for Median Regression}

In \cref{tab:median}, we report the average mean absolute error (MAE) on test data for median regressions using the same experimental setup as in our main paper.
We observe similar trends as for mean regression and probabilistic forecasting. Our random search methods outperform the greedy \texttt{Dissmann} and \texttt{Kraus} algorithms, with \texttt{RS-E} generally performing best.

\begin{table}[!t]\small
\caption{Rounded test MAE (standard error) for median regressions. Best results in bold.}
\label{tab:median}
\begin{tabular*}{\linewidth}{@{\extracolsep{\fill}}lllllll}
\toprule
 & Energy & Concrete & Airfoil & Wine & Ccpp & California \\ 
\midrule\addlinespace[2.5pt]
\texttt{Dissmann} & 1.85 (0.06) & 5.38 (0.10) & 3.31 (0.05) & 0.43 (0.01) & 3.15 (0.01) & 0.46 (0.00) \\
\texttt{Kraus} & 1.93 (0.08) & 5.55 (0.08) & 3.30 (0.06) & 0.44 (0.01) & 3.32 (0.01) & 0.45 (0.00) \\
RS-B (50) & 1.56 (0.11) & 5.35 (0.12) & 2.71 (0.04) & 0.43 (0.01) & 3.14 (0.02) & 0.43 (0.00) \\
RS-B (100) & 1.59 (0.20) & 5.30 (0.11) & 2.75 (0.05) & 0.43 (0.01) & 3.13 (0.02) & 0.42 (0.00) \\
RS-B (500) & 1.32 (0.10) & 5.28 (0.12) & \textbf{2.69 (0.04)} & 0.42 (0.01) & 3.15 (0.02) & 0.42 (0.00) \\
RS-E (50) & 1.15 (0.04) & 4.74 (0.09) & 2.69 (0.04) & \textbf{0.40 (0.00)} & 3.13 (0.02) & 0.40 (0.00) \\
RS-E (100) & 1.15 (0.04) & 4.71 (0.08) & 2.75 (0.05) & 0.40 (0.00) & 3.13 (0.02) & 0.41 (0.00) \\
RS-E (500) & \textbf{1.05 (0.02)} & \textbf{4.70 (0.07)} & 2.75 (0.05) & 0.40 (0.00) & \textbf{3.12 (0.01)} & \textbf{0.40 (0.00)} \\
\bottomrule
\end{tabular*}

\end{table}

\subsection{Runtime Comparison on Other Datasets}

In \cref{fig:cpu_time_all}, we show the relative CPU time of our random search methods compared to \texttt{Dissmann} on all data sets.
As expected, the CPU time grows roughly linearly with $M$ in training, irrespective of the random search method or the task.
On the other hand, there is no increase in inference time for \texttt{RS-B}, since only one model is used for prediction.
For \texttt{RS-E}, the inference time can grow with $M$, since some fraction of the candidates are retained in the MCS and used for prediction.
However, this increase is not systematic, as there can sometimes be a single (or a few) model(s) in the MCS, irrespective of $M$.

In \cref{tab:cpu_time}, we report the median CPU time for training and inference using a single vine on all data sets.
Except for the \texttt{California} data set, training and inference using a single vine is very fast, taking at most around one second for training and less than 0.5 seconds for inference on a single CPU core.
Note that, for inference, the reported CPU time is for predicting the mean, median, and 101 quantiles (from 0\% to 100\% with a step of 1\%) of the conditional distribution of the response given the features at all test points.
With around 20'000 data points, the \texttt{California} data set is the largest one considered in our experiments, and for such sizes we would need to optimize how pair-copulas are fitted to further reduce the computational cost.
This can be done, for instance, by using a smaller \texttt{TLL} grid size initially (i.e., before applying the MCS), and then re-fitting the selected candidates with a larger grid size.
\begin{figure}[t]
    \centering
    \includegraphics[width=1\linewidth]{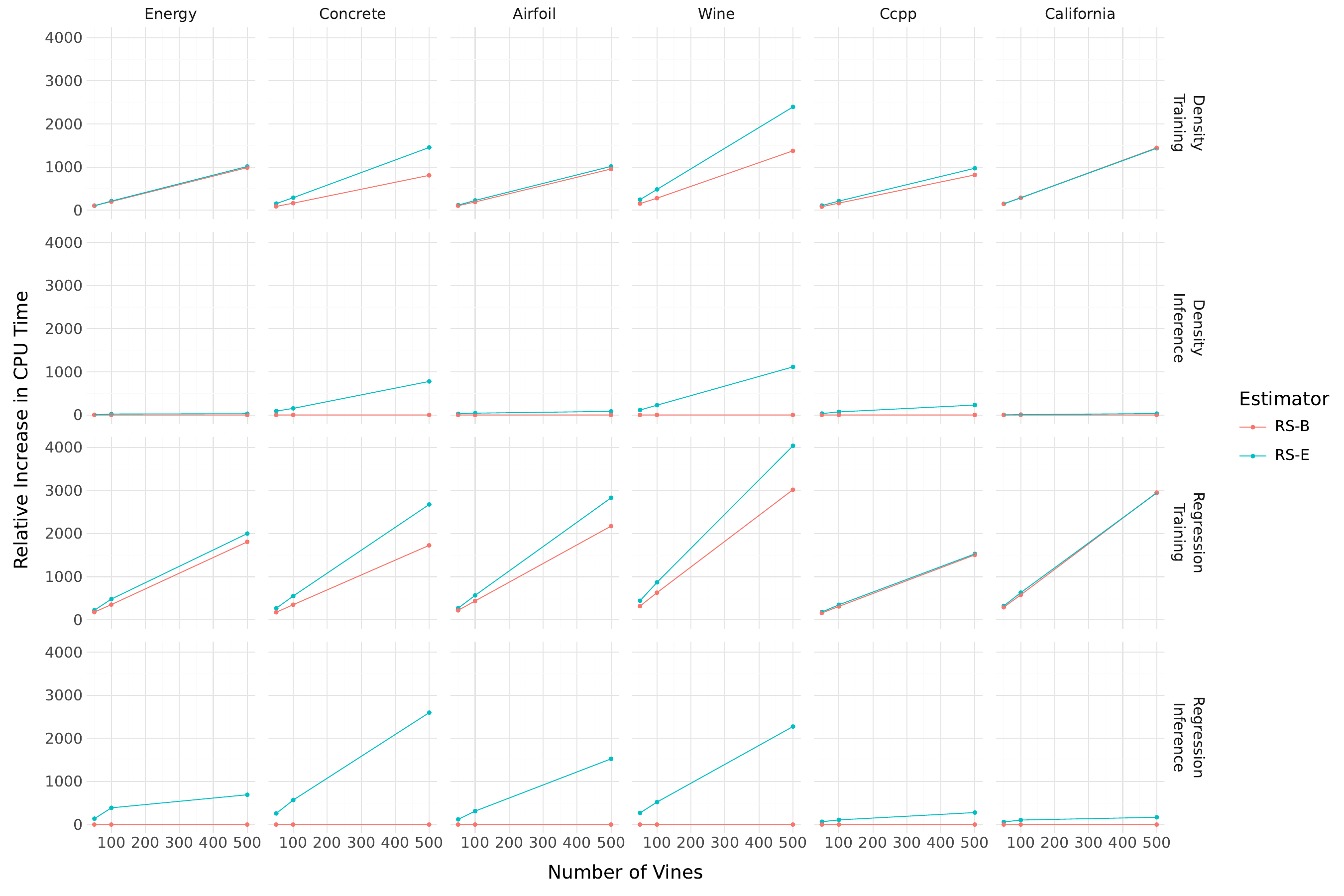} 
    \caption{Relative CPU time of random search methods compared to \texttt{Dissmann} on all data sets.}
    \label{fig:cpu_time_all}
\end{figure}

\begin{table}[!t]\small
\caption{Median CPU time (in seconds) for training and inference using a single vine.}
\label{tab:cpu_time}
\begin{tabular*}{\linewidth}{@{\extracolsep{\fill}}lrrrrrr}
\toprule
 & \multicolumn{2}{c}{Energy} & \multicolumn{2}{c}{Concrete} & \multicolumn{2}{c}{Airfoil}  \\ 
\cmidrule(lr){2-3} \cmidrule(lr){4-5} \cmidrule(lr){6-7} 
 & Training & Inference & Training & Inference & Training & Inference \\ 
\midrule\addlinespace[2.5pt]
Density & 0.311 & 0.005 & 0.367 & 0.006 & 0.204 & 0.004  \\
Regression & 0.382 & 0.33 & 0.415 & 0.453 & 0.213 & 0.313 \\
\midrule\addlinespace[2.5pt]
 & \multicolumn{2}{c}{Wine} & \multicolumn{2}{c}{Ccpp} & \multicolumn{2}{c}{California}  \\ 
\cmidrule(lr){2-3} \cmidrule(lr){4-5} \cmidrule(lr){6-7} 
 & Training & Inference & Training & Inference & Training & Inference \\ 
\midrule\addlinespace[2.5pt]
Density  & 0.644 & 0.015 & 0.742 & 0.01 & 6.141 & 0.047 \\
Regression & 0.659 & 1.141 & 1.011 & 1.531 & 5.109 & 7.995 \\
\bottomrule
\end{tabular*}

\end{table}

\subsection{Generating vines by local perturbations} \label{sec:wilson}

In this appendix, we report additional experiments comparing uniform random search to a locally perturbed version of Dissmann's algorithm.
Specifically, instead of selecting the maximum spanning tree (MST) obtained from the empirical Kendall's $\tau$ matrix at each tree level, we generate random trees according to distribution
\begin{align*}
  P(T) \propto \prod_{(j,k) \in T} |\tau_{j,k}|,
\end{align*}
with the MST being the mode.
Uniform sampling from this distribution can be done efficiently via Wilson's loop-erased random walk algorithm \citep{wilson1996generating}, implementented e.g. in \texttt{boost} \citep{boost} and made available in the vine selection context by \texttt{pyvinecopulib} \citep{pyvinecopulib}.

\begin{table}[th]\small
\centering
\caption{Additional results for the experiment in \cref{tab:NLL} for a variant of \cref{alg:MCS} that generates vines by sampling local perturbations of the Dissmann structure.}
\label{tab:local-perturbation-baseline}
\begin{tabular}{lcccccc}
\toprule
 & Energy & Concrete & Airfoil & Wine & Ccpp & California \\
\midrule
\texttt{RS-B (50)}  & 1.54 (0.28)  & 7.1 (0.17)  & 3.11 (0.06) & 8.47 (0.27) & 6.77 (0.01) & 3.56 (0.28) \\
\texttt{RS-B (100)} & 1.17 (0.40)   & 7.11 (0.16) & 3.1 (0.07)  & 8.47 (0.27) & 6.77 (0.01) & 3.57 (0.28) \\
\texttt{RS-B (500)} & 0.48 (0.67)  & 6.99 (0.18) & 3.11 (0.06) & 8.47 (0.27) & 6.76 (0.01) & 3.58 (0.28) \\
\texttt{RS-E (50)}  & 0.53 (0.37)  & 6.63 (0.13) & 3.06 (0.07) & 8.19 (0.33) & 6.75 (0.01) & 3.51 (0.26) \\
\texttt{RS-E (100)} & 0.19 (0.33)  & 6.61 (0.12) & 3.06 (0.07) & 8.09 (0.33) & 6.75 (0.01) & 3.49 (0.27) \\
\texttt{RS-E (500)} & -0.11 (0.36) & 6.57 (0.15) & 3.07 (0.07) & 7.8 (0.21)  & 6.74 (0.01) & 3.46 (0.27) \\
\bottomrule
\end{tabular}
\end{table}

\cref{tab:local-perturbation-baseline} reports results for the experiment in \cref{tab:NLL} of the main text. 
We see that uniform random search slightly outperforms the local perturbation.
From a theoretical standpoint, there is no reason to believe that a greedy approach will find a good structure. In fact, it is easy to simulate data where greedy methods (and perturbed variants) fail catastrophically: construct a vine where the first tree consists of weak dependencies, while the later trees contain strong dependencies and asymmetric relationships. The greedy method will pick a structure that severely violates the simplified vine assumption (see Section 2.3), and many randomly sampled structures lead to better performance. However, alternative local methods may be interesting to explore in future work.




\end{document}